\documentclass[11pt]{article}

\usepackage[margin=1in]{geometry}
\usepackage{setspace}
\usepackage{float}
\usepackage{array, makecell} %
\usepackage{xcolor,colortbl}
\usepackage{framed,color}
\usepackage{xcolor}

\usepackage{qcircuit}
\usepackage{tikz}
\usetikzlibrary{calc,arrows.meta,decorations.pathmorphing,fit}
\usepackage{yquant}
\useyquantlanguage{groups}
\usetikzlibrary{quotes}

\newcommand{\controlu}{*-=[][F]{\phantom{\bullet}}}

\usepackage{authblk}

\usepackage{amsthm,amsmath,amssymb}
\usepackage{bbm}
\usepackage[bookmarks=true,hypertexnames=false,pagebackref]{hyperref}
\hypersetup{colorlinks=true, citecolor=blue, linkcolor=red,
  urlcolor=blue}
\usepackage{newpxtext}
\usepackage[libertine,vvarbb]{newtxmath}

\usepackage[T1]{fontenc} 
\usepackage{textcomp} 
\usepackage{blkarray}   
\usepackage[scr=rsfso]{mathalfa}
\usepackage{nicefrac}
\usepackage{cleveref}
\usepackage{graphicx}
\usepackage{aliascnt}
\usepackage[section]{placeins}
\usepackage{tcolorbox}
\usepackage{epigraph}
\usepackage{xifthen}
\usepackage{latexsym}
\usepackage{framed}
\usepackage{fullpage}
\usepackage{bm}
\usepackage{braket}

\usepackage{thmtools,thm-restate}

\usepackage{algorithm}
\usepackage{algorithmic}

\newtheorem{theorem}{Theorem}[section]

\newtheorem{fact}[theorem]{Fact}
\newtheorem{corollary}[theorem]{Corollary}
\newtheorem{lemma}[theorem]{Lemma}
\newtheorem{definition}[theorem]{Definition}

\newtheorem{example}[theorem]{Example}

\newtheorem{remark}[theorem]{Remark}
\newtheorem*{remark*}{Remark}

\usepackage{multicol}
\usepackage{multirow}

\usepackage{mleftright}

\newcommand{\CC}{\mathcal{C}}
\newcommand{\QFT}{\mathrm{QFT}}

\renewcommand{\kappa}{\ell}

\newcommand{\poly}{{\sf poly}}

\newcommand{\warn}[1]{{ \color{red}#1} }

\newcommand{\rank}{\mathrm{rank}}

\newcommand{\sdotfill}{\textcolor[rgb]{0.8,0.8,0.8}{\dotfill}} 

\title{
  Communication Complexity of Distributed Unitary Synthesis
}

\date{}
\author[1]{Longcheng Li \thanks{lilongcheng116@gmail.com} }
\author[1,3]{Xiaoming Sun  \thanks{sunxiaoming@ict.ac.cn } $^{\mathparagraph}$ }
\author[2,3]{Jialin Zhang \thanks{zhangjialin@ict.ac.cn} }
\author[1]{Jiadong Zhu \thanks{zhujiadong2016@163.com} $^{\mathparagraph}$ }

\affil[1]{State Key Lab of Processors, Institute of Computing Technology, Chinese Academy of Sciences, Beijing 100190, China}
\affil[2]{Institute of Computing Technology, Chinese Academy of Sciences, Beijing 100190, China}
\affil[3]{School of Computer Science and Technology, University of Chinese Academy of Sciences, Beijing 100049, China}

\begin{document}
\maketitle

\renewcommand{\thefootnote}{\textparagraph}  
\footnotetext{Corresponding authors.}  
\renewcommand{\thefootnote}{\arabic{footnote}}

\begin{abstract}

We study space-bounded communication complexity for unitary implementation in distributed quantum processors, where we restrict the number of qubits per processor to ensure practical relevance and technical non-triviality. We model distributed quantum processors using distributed quantum circuits with nonlocal two-qubit gates, defining the distributed communication complexity of a unitary as the minimum number of such nonlocal gates required for its realization, up to permutations of data qubit positions.

Our contributions are twofold. First, for general \(n\)-qubit unitaries, we improve upon the trivial \(O(4^n)\) communication bound. With \(k\) pairwise-connected processors (each having \(n/k\) data qubits and \(m\) ancillas), we prove the communication complexity satisfies \(O\left(\max\{4^{(1-1/k)n - m}, n\}\right)\) —e.g., \(O(2^n)\) when \(m=0\) and \(k=2\)—and show this upper bound is tight. We also extend the analysis to approximation and general topology models. Second, for special unitaries, we demonstrate that both Quantum Fourier Transform (QFT) and Clifford circuits admit linear upper bounds on communication complexity within the exact model. This outperforms the trivial quadratic communication complexity that applies to these two unitaries. In the approximation model, QFT’s communication complexity reduces drastically from linear to logarithmic, while Clifford circuits still require at least linear communication. These results provide fundamental insights into optimizing communication for distributed quantum unitary implementation, advancing the feasibility of large-scale DQC systems.
\end{abstract}

\section{Introduction}

Communication Complexity is a foundational research topic in distributed quantum computing (DQC), dedicated to quantifying the minimum number of qubits that need to be transmitted to complete distributed computing tasks \cite{10.1145/1250790.1250866,10.1145/3519935.3519954,10.1145/3717823.3718155,10.1145/3732772.3733509,li2025device}. DQC encompasses two core cost components: intra-processor computing and inter-processor communication \cite{denchev2008distributed,sundaram2024distributing,cuomo2025architectures}. While intra-processor computing is already computationally intensive and physically challenging to implement, inter-processor communication introduces even more severe bottlenecks—positioning the optimization of the communication complexity as a central priority in DQC research \cite{pant2019routing,li2022fidelity,sundaram2022distribution,victora2023entanglement,gong2024optimal}.

In quantum information science, the study of the communication complexity is frequently tied to two core motivations. For classical computing tasks, a key focus is investigating the gap between quantum and classical communication complexity—more precisely, whether qubits can achieve identical computational objectives with lower communication overhead than classical bits  \cite{10.1145/333623.333628,10.1145/1250790.1250866,10.1145/2698587,10.1145/3519270.3538441,10.1145/3625225}. Most research in this domain does not constrain the memory of local processors (in both classical and quantum communication settings)  \cite{10.5555/2634074.2634210,10.1145/3357713.3384243,10.1145/3618260.3649604,10.1145/3717823.3718155}, rendering superlinear communication lower bounds unachievable: a processor could simply transmit its entire input, evading meaningful communication optimization. Here, communication complexity stems from the information gap where one processor lacks access to the inputs of other processors. For inherently quantum tasks, such as remote preparation of quantum states \cite{le2023distributed, chen2024preparing} and distributed implementation of quantum algorithms \cite{yimsiriwattana2004distributed, parekh2021quantum, feng2024distributed}, however, the limitations of current quantum hardware drive the need for distributed solutions. Contemporary quantum processors are restricted by their qubit capacity, making DQC—which integrates multiple quantum processors—a vital strategy for tackling large-scale quantum tasks beyond the capabilities of single-processor systems \cite{boschero2025distributed}. Importantly, aligned with the hardware constraint of limited qubits per processor, this branch of study incorporates a space bound: it seeks not only to minimize communication complexity but also to limit the number of qubits stored in each individual processor. In this line of research, communication complexity primarily arises not from the information gap, but from the intrinsic complexity of space-bounded quantum circuits.

This paper addresses the second motivation of communication complexity research.  
We model DQC using distributed quantum circuits—collections of quantum circuits that support nonlocal two-qubit gates, which act on data qubits belonging to different processors.  
In practice, such nonlocal gates can be implemented either directly~\cite{song2024realization} or indirectly via quantum teleportation~\cite{long2022interfacing,liu2024nonlocal,main2025distributed} or cat-entangler/disentangler constructions~\cite{yimsiriwattana2004distributed}.  
Since these realization methods differ only by a constant factor in communication cost, the distinction is immaterial to the asymptotic analysis.  
Because quantum computation can fundamentally be modeled as the application of a unitary transformation to a quantum system, we study how to implement a given unitary operation in a distributed setting while minimizing the associated communication complexity.  
Specifically, we consider an \(n\)-qubit unitary \(U\) realized across \(k\) processors, each equipped with \(n/k\) data qubits
and \(m\) ancilla qubits ($m \leq n$). We define the \emph{distributed communication complexity} (DCC) \footnote{
    As our model allows choosing the partition of data qubits,
    we name the cost measure as \emph{distributed} communication complexity to
    distinguish it from the conventional communication complexity where the input partition is given as part of the problem definition.
    }
of \(U\) as the minimum number of nonlocal two-qubit gates\footnote{While communication complexity is often measured by the number of exchanged qubits, the number of nonlocal two-qubit gates provides a more natural measure in the circuit model. Both metrics are equivalent up to a constant factor.} required for the \(k\) processors to implement \(U\), where optimization is performed over all permutations of data qubit positions. This formulation differs critically from conventional communication complexity, in which the input partition between parties is fixed a priori and cannot be adjusted. In our model, we do not restrict qubits to fixed processing units; instead, we aim to minimize genuine communication overhead by allowing arbitrary balanced partitions. By contrast, canonical problems in standard communication complexity would become degenerate if arbitrary partitions were permitted. For example, the classic disjointness problem becomes trivial if Alice and Bob may freely rearrange the input bits between them, rather than following a predetermined split. Our definition thus captures the intrinsic communication cost of distributed unitary implementation, independent of artificial constraints on qubit placement (see \Cref{sec:formaldefs} for formal definitions).

First, we study the distributed communication complexity of \emph{general \(n\)-qubit unitaries}, a topic of primary theoretical interest. A fundamental result in quantum circuit synthesis shows that implementing an arbitrary \(n\)-qubit unitary requires \(\Theta(4^n)\) two-qubit gates in the non-distributed setting~\cite{shende2004minimal}. Extending this to distributed quantum systems, the straightforward upper bound on communication complexity remains \(O(4^n)\) when the number of ancilla qubits \(m\) is limited (e.g., when \(m\) is sublinear in \(n\)). The central question is whether this bound can be improved and, more generally, how to characterize the trade-off between communication complexity and the available ancilla resources.

Second, we analyze the distributed communication complexity of two important classes of unitaries: \emph{the Quantum Fourier Transform (QFT) and Clifford circuits}. The QFT is a fundamental component of many quantum algorithms, such as Shor’s factoring algorithm \cite{shor1994algorithms} and quantum phase estimation \cite{kitaev1996quantum}, whereas Clifford circuits form the backbone of quantum error correction \cite{lidar2013quantum}. Understanding their communication cost in distributed settings is therefore crucial for developing large-scale, fault-tolerant quantum computing systems.
Both QFT and Clifford circuits require a quadratic number of two-qubit gates in their standard, non-distributed realizations. Under the constraint of limited ancilla qubits (e.g., when $m$ is constant), directly partitioning these circuits across processors leads to a trivial quadratic communication cost. Our goal is to refine this result to provide a tight characterization of their communication complexity in distributed settings, thereby enabling more efficient implementations in practice.

When analyzing distributed communication complexity, we consider two key model
variations: unitary implementations may be either exact or approximate, and the
connection topology of processors may be either a complete (pairwise-connected)
graph or an arbitrary connected graph. For all unitaries studied, our analysis begins with the
exact implementation model under complete connectivity. Then we further extend
the results to the approximate model, the general topology model, or both.

\subsection{Contribution and Outline}

In response to these research questions, we make two primary contributions to the field, as summarized in \Cref{tab:main}.

\begin{itemize}
    \item Distributed Communication Complexity of General Unitaries: 
    \begin{itemize}
        \item We improve upon the trivial \(O(4^n)\) bound by proving: for pairwise-connected processors (where any pair of processors can execute nonlocal two-qubit gates), 
        any $n$-qubit unitary can be implemented with communication complexity  \(O\left(\max\{4^{(1-1/k)n - m}, n\}\right)\). Notably, we provide a matching lower bound, i.e., there exists an $n$-qubit unitary that 
        requires \(\Omega\left(\max\{4^{(1-1/k)n - m}, n\}\right)\) communication, 
        showing that our construction is tight. To illustrate the improvement over the trivial bound, consider the case when \(m=0\) and \(k=2\): our result reduces the communication complexity to \(\Theta(2^n)\), a quadratic improvement compared with \(O(4^n)\). 
        \item  In the approximate model, to implement any $n$-qubit unitary within error $\epsilon$, we prove that communication complexity is lower bounded by $\Omega\left(\max\left\{4^{(1-1/k)n-m} \frac{\log(1/\epsilon)}{n}, n\right\}\right)$, which differs from that of the exact case by a linear factor. 
        
        \item For the general topology model\footnote{In the general
        topology model, processors are connected by an arbitrary
        connected graph $G$, and a nonlocal two-qubit gate is allowed only
        between two processors connected by an edge of $G$.}, we prove an upper
        bound of $O(\max\{4^{(1-1/k)n-m},Dn\})$, where the interconnect topology
        between processors is described by a graph \(G\) with diameter \(D\). An
        interesting interpretation of this result is that relaxing the topology
        from pairwise-connected to general does not significantly blow up the
        communication: the diameter \(D\) of the topology graph only impacts the
        bound when the typical leading term $4^{(1-1/k)n-m}$ is very small—and
        this scenario arises only in unrealistic cases where \(m\) is nearly as
        large as \(n\). 
    \end{itemize} 
     \item Distributed Communication Complexity of Special Unitaries: 
\begin{itemize}
    \item In the exact model, we improve upon the trivial quadratic bound by showing that both the \(n\)-qubit QFT and Clifford circuits can be implemented with \(O(kn)\) communication complexity among \(k\) pairwise-connected processors when each processor has one ancilla qubit (\(m=1\)). For the two-processor case, this reduces to \(O(n)\). We further establish a matching \(\Omega(n)\) lower bound that holds even with unbounded ancilla qubits, demonstrating that our construction is tight for two processors. 
    \item In the approximate model, we show that allowing a small implementation error \(\epsilon\) in QFT leads to a significant reduction in communication to \(O(k\log (n/\epsilon))\) while maintaining \(m=1\). In the two-processor case, this yields an exponential improvement from \(\Theta(n)\) to \(O(\log (n/\epsilon))\). In contrast, there exists a Clifford circuit which still requires at least linear communication, even when approximation is permitted. 
    \item In the general topology model, the QFT can be implemented over any connected processor graph with \(O(kn)\) communication, incurring no additional topology overhead. In contrast, Clifford circuits require \(O(Dkn)\) communication, where \(D\) is the diameter of the topology graph.
\end{itemize}

\end{itemize}

\begin{table}[t]
\centering
\caption{Summary of communication complexity of different unitaries \label{tab:main}}

\renewcommand{\arraystretch}{1.25}
\setlength{\tabcolsep}{4.5pt}
\small
\begin{tabular}{|c||c||c|c|c|c|}
\hline
\multirow{2}{*}{\makecell[c]{\textbf{Unitary}\\\textbf{Type}}} 
& \multirow{2}{*}{\textbf{Previous Known}} 
& \multicolumn{3}{c|}{\textbf{Pairwise-Connected}} 
& \multirow{2}{*}{\makecell[c]{\textbf{General Topology}\\\textbf{(Exact)}}} \\ 
\cline{3-5}
& & \multicolumn{2}{c|}{\textbf{Exact}} & \textbf{Approximate} & \\ 
\hline

General
& $O(4^n)$ (Trivial) 
& \multicolumn{2}{c|}{$\Theta\!\left(\max\{\mathcal{A},\, n\}\right)$} 
& $\Omega\!\left(\max\{\mathcal{A}\tfrac{\log(1/\epsilon)}{n},\, n\}\right)$ 
& $O\!\left(\max\{\mathcal{A},\, Dn\}\right)$ \\ 
\hline

QFT
& \makecell[c]{$k=2$: $O(n)$ \cite{neumann2020imperfect}\\
$k>2$: $O(n^2)$ (Trivial)} 
& $O(kn)$ 
& $\Omega(n)$ 
& $O\!\left(k\log(n/\epsilon)\right)$ 
& $O(kn)$ \\ 
\hline

Clifford
& $\tilde{O}(n^2)$ (Trivial) 
& $O(kn)$ 
& $\Omega(n)$ 
& $\Omega(n)$ 
& $O(Dkn)$ \\ 
\hline
\end{tabular}

\vspace{1mm}

\begin{minipage}{0.91\linewidth}
\footnotesize
$\mathcal{A} = 4^{(1-1/k)n - m}$;
$n$ — number of data qubits; 
$k$ — number of processors;
$m$ — number of ancilla qubits per processor;
$D$ — diameter of the topology graph;
$\epsilon$ — approximation error in spectral norm.

\textbf{a)} ``Trivial'' means the direct partitioning of the corresponding non-distributed circuit.
\textbf{b)} For QFT and Clifford, the upper bound results hold for arbitrarily small $m\geq 1$, while the lower bound results hold for arbitrarily large $m$.

\end{minipage}
\end{table}

The paper is organized as follows: In Section \ref{sec:relatedwork}, we provide
a detailed overview of related work; in Section \ref{sec:discussion}, we briefly
outline several future directions for this research; in \Cref{sec:tech}, we give
an overview of the key techniques used in our proofs. Section \ref{sec:prelim}
introduces the notation and background concepts from quantum information
science. In Section \ref{sec:formaldefs}, we present the formal definitions of
the distributed communication complexity of unitaries, including its
approximate and multi-party variants. Sections \ref{sec:commk}, \ref{sec:QFT},
and \ref{sec:Clifford} are devoted to studying the distributed communication complexity of
general unitary operations, QFT, and Clifford circuits, respectively.

\subsection{Related Work}\label{sec:relatedwork}

Communication complexity theory, initiated by Andrew Yao \cite{10.1145/800135.804414}, stands as a central branch of theoretical computer science. Since Yao’s foundational work, researchers have investigated the communication complexity of numerous fundamental functions—including Set Disjointness \cite{RAZBOROV1992385,6686203}, Inner Product \cite{10.1145/3428671}, Voting \cite{10.5555/2343776.2343781}, Submodular Maximization \cite{10.1145/3588564}, and other problems \cite{10.1145/2746539.2746560,10.1145/2684464.2684470,10.1145/3055399.3055407,10.1145/3406325.3451083,li_et_al:LIPIcs.DISC.2024.32}—as well as across various models, including randomized, non-deterministic, and fault-tolerant models \cite{10.1145/800061.808737, 10.5555/3381089.3381194,10.1145/3519935.3520078,10.1145/3618260.3649604}. Concurrently, a suite of key analytical techniques has been developed to bound communication complexity, such as the fooling set method, the rank bounds method, and the probabilistic method \cite{10.1145/800070.802208,4568106,PATURI1986106,10.1145/276698.276883,10.1016/j.jcss.2003.11.006,doi:10.1137/080733644}. Complementing these efforts, ongoing work has focused on advancing the general theory of communication complexity and formalizing its complexity classes \cite{10.1145/800061.808742,4568225,10.1145/2422436.2422456,10756137}.


The communication complexity of implementing standard quantum algorithms in a
distributed manner has been widely investigated in the literature. Examples
include distributed versions of Shor's algorithm
\cite{yimsiriwattana2004distributed,xiao2022distributed}, Grover’s search
\cite{qiu2024distributed}, Simon’s algorithm \cite{tan2022distributed}, quantum phase estimation \cite{neumann2020imperfect}, and
quantum simulation \cite{feng2024distributed}. These studies primarily
aim to design distributed implementations of specific algorithms with low
inter-processor cost communication. In \cite{neumann2020imperfect}, the authors design a distributed QFT that achieves linear communication across two processors as a subroutine, but they neither provide lower-bound analyses nor extend their study to the multi-processor setting. A concurrent work
\cite{ebnenasir2025lower} derives a lower bound on the communication cost of
distributed QFT, but its model differs from
ours in that it does not permit general circuit transformations.

A separate line of research investigates the communication cost of distributed
quantum circuits in a non-asymptotic setting. In this framework, a given
distributed circuit is treated as input, and the goal is to reduce its
communication cost through heuristic or compiler-level optimization techniques
\cite{zomorodi2018optimizing,andres2019automated,dadkhah2021new,g2021efficient,sundaram2022distribution,wu2023qucomm}.
These approaches rely mainly on empirical evaluation to demonstrate
effectiveness and do not provide formal performance guarantees. 

\subsection{Open Questions and Discussions}\label{sec:discussion}

In this work, we establish a comprehensive framework for analyzing the communication complexity of distributed unitary implementation and derive several tight bounds on the distributed communication complexity of (i) general \(n\)-qubit unitaries and (ii) the QFT and Clifford circuits. Despite these advances, several important gaps remain open.

For general \(n\)-qubit unitaries, we have shown a tight \(\Theta\!\left(\max\{4^{(1-1/k)n - m},\, n\}\right)\) communication complexity in the exact model under fully connected processor topology. In the approximate model, however, the best known lower bound reduces to \(\Omega\!\left(\max\{4^{(1-1/k)n - m}\tfrac{\log(1/\epsilon)}{n},\, n\}\right)\), leaving a linear gap when error $\epsilon\geq 1/\mathrm{poly}(n)$ and the number of ancilla qubits \(m\) is limited. It would be valuable to determine whether this lower bound is tight or can be strengthened by removing the \(1/n\) term.  
For general topologies, the upper bound increases to \(O(\max\{4^{(1-1/k)n - m},\, Dn\})\), where \(D\) denotes the diameter of the topology graph. An open question is whether \(\Omega(Dn)\) communication is indeed necessary when unbounded ancilla qubits are available.

For QFT and Clifford circuits, we have established a tight \(\Theta(n)\) communication complexity in the exact two-processor setting. For \(k\) processors, the current upper bound generalizes to \(O(kn)\), but no matching lower bound is known. Current proof techniques, based on rank methods or information-theoretic arguments, fail to derive superlinear lower bounds. Another direction is to determine whether the \(O(\log(n/\epsilon))\)-communication QFT construction across two processors is optimal in the approximate model. Interestingly,  \cite{chen2023quantum} shows that the QFT without bit reversal generates only a constant amount of entanglement, suggesting the possibility of approximately implementing QFT with even lower, potentially constant, communication cost. 

Finally, we expect the proposed framework to be broadly applicable to the study of communication complexity in other quantum computing tasks—such as Uhlmann transformation~\cite{bostanci2023unitary}, quantum state preparation~\cite{sun2023asymptotically,chen2026scalable}, and quantum simulation~\cite{feng2024distributed}—under distributed settings. More broadly, we hope that this work lays the foundation for understanding the communication complexity of inherently quantum tasks and its trade-off with local quantum memory in distributed quantum computation.

\begin{figure}
    \centering
    \includegraphics[width=\textwidth]{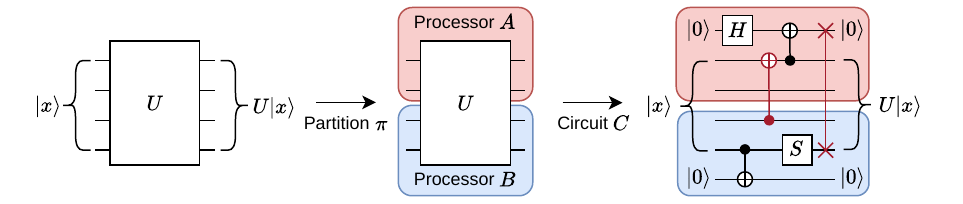}
    \caption{
    Illustration of the model setting for $k=2$ processors and $m=1$ ancilla.
    (i) The leftmost is the target unitary $U$ with $n=4$ input qubits.
    (ii) The middle is $U$ placed on two processors $A$ and $B$ by choosing a balanced partition $\pi$ on input qubits.
    (iii) The rightmost is an example quantum circuit $C$ that implements $U$ with one clean ancilla qubit on each processor.
    The communication cost of $C$ is $2$ as there are two nonlocal two-qubit gates (marked red).
    The distributed communication complexity of $U$ with ancilla bound $m$ is minimum number of nonlocal two-qubit gates optimized over all possible $\pi$ and $C$.
    } \label{fig:main}
\end{figure}

\subsection{Technical Overview} \label{sec:tech}

This section overviews the techniques behind our proofs. For clarity, we focus on the two-processor case ($k=2$), where the topology graph can be ignored. The main difficulties we overcome are threefold.

\textbf{Upper bounds: synthesis for communication.} Traditional circuit synthesis methods (e.g., quantum Shannon decomposition, PLU decomposition) were designed for non‑distributed settings, where the goal is to minimize circuit size or depth; how to adapt them to optimize communication cost was not known. At the same time, most prior work on distributed quantum computation simply takes a fixed circuit and cuts it across processors---an easy approach that avoids redesigning the circuit but often yields poor communication. Thus, neither line of work addresses the important but poorly understood problem of \emph{synthesizing a circuit from scratch with communication as the primary objective}. We show that this can be done by rethinking the synthesis process: for general unitaries, we recursively apply a decomposition that limits nonlocal interactions to a small fraction of qubits at each level, leading to an quadratic reduction in communication compared to the trivial bound; for the QFT, we reorganize the circuit so that exact implementation uses only linear communication, and by allowing a small approximation error we further reduce it to logarithmic; for Clifford circuits, we factor the transformation into simple blocks that can each be executed with linear communication. These innovations show that synthesis can achieve far lower communication than simply cutting a known circuit, a direction that previous distributed studies did not pursue.

\textbf{Communication complexity under bounded ancillas.} 
Most prior communication complexity studies allow unlimited local computing power (arbitrarily many ancillas). In that regime, the communication cost of any unitary is at most linear in the number of qubits, and many clean information‑theoretic lower bounds become possible. However, when we restrict each processor to only $m$ ancilla qubits, the landscape changes dramatically—for both upper and lower bounds. On the upper bound side, the limited ancillas make it nontrivial to design efficient protocols: we can no longer freely store or buffer intermediate states across processors. We overcome this by carefully packing input qubits using only $O(n)$ nonlocal SWAPs, then applying a synthesis procedure whose communication cost is $O(2^{n-2m})$. This cost becomes exponential when $m$ is small and decreases rapidly as $m$ increases. On the lower bound side, bounded ancillas also invalidate standard information‑theoretic arguments (which rely on unlimited local resources). We extend a classical technique—parameter counting combined with covering numbers for polynomial images—to the bounded‑ancilla setting, yielding a matching exponential lower bound. Thus, bounding ancillas forces communication to become exponential (in contrast to the linear cost under unlimited ancillas), and establishing tight bounds in this regime requires new algorithmic and proof techniques on both sides.

\textbf{Lower bounds universal over all balanced partitions.} Our lower bounds must hold for \emph{every} balanced partition of the $n$ input qubits. There are exponentially many such partitions ($\binom{n}{n/2}$), which is far more demanding than a fixed-partition bound. We overcome this by constructing specific unitaries that are hard regardless of how the qubits are split: the shifted inner-product unitary for linear lower bound, and random invertible matrices for Clifford circuits. These constructions force $\Omega(n)$ communication under any balanced partition, and the proof requires new combinatorial arguments to argue about all partitions simultaneously.

Below, we elaborate on the proof ideas for general unitaries, the quantum Fourier transform, and Clifford circuits.

\paragraph{Model setting.} As in 
\Cref{fig:main}, we aim to implement an $n$-qubit unitary 
$U$ using two quantum processors $A$ and $B$. We may choose a balanced partition $\pi$ of
the $n$ input qubits between the two processors, so that each processor holds
$n/2+O(1)$ qubits. In addition, each processor is provided with $m$ clean
ancilla qubits to model the local space constraint. A quantum
circuit $C$ on $(n+2m)$ qubits implements (resp., approximates) $U$ if, on the
input qubits, it acts exactly (resp., approximately) as $U$. We define the
communication cost of $C$ to be the number of \emph{nonlocal} two-qubit gates, namely,
two-qubit gates whose endpoints lie on different processors. The \emph{distributed communication
complexity of $U$ with ancilla bound $m$} is the minimum communication cost over all balanced
partitions $\pi$ and all circuits $C$.

\subsubsection{General Unitaries}
\paragraph{Upper bound.} Implementing an arbitrary $n$-qubit unitary on a single processor requires
$\Theta(4^n)$ two-qubit gates \cite{shende2004minimal}, which gives a trivial
$O(4^n)$ bound on the distributed communication complexity. Our first 
result shows that one can do substantially better. 

\begin{theorem}[Informal $k=2$ version of \Cref{lem:upperk}]\label{thm:inf-general-upper}
Any $n$-qubit unitary $U$ can be implemented with
$O(\max\{2^{n-2m},n\})$ nonlocal two-qubit gates.
\end{theorem}

We first design a synthesis algorithm with no ancillas and $O(2^n)$ communication. The main idea is to adapt the quantum Shannon decomposition (QSD) to the distributed setting. A single QSD step
decomposes an $n$-qubit unitary $U$ into four $(n-1)$-qubit unitaries
interleaved with three $n$-qubit uniformly controlled rotation (UCR) gates.
Applying QSD recursively for $t$ levels yields:
(1) $4^t$ local $(n-t)$-qubit unitaries acting locally on Processor~$B$; and
(2) for each level $i\in\{0,1,\ldots,t-1\}$, exactly $3\cdot 4^i$ UCRs on
$(n-i)$ qubits, with the cut after the $(t-i)$-th qubit.

The key ingredient is the following lemma.
\begin{lemma}[Informal version of \Cref{lem:ucdecomp}]
An $n$-qubit UCR can be implemented using $O(2^t)$ nonlocal two-qubit gates
when Processor~$A$ holds $t$ qubits and Processor~$B$ holds the remaining
$n-t$ qubits.
\end{lemma}
Therefore, each UCR at level $i$ costs $O(2^{t-i})$ 
communication, so the total communication is
\(
\sum_{i=0}^{t-1} 3\cdot 4^i \cdot O(2^{t-i}) = O(2^{2t}).
\)
Setting $t=n/2$ gives an $O(2^n)$-communication implementation with no
ancillas.
For the general case $m\ge 0$, we first use $O(n)$ nonlocal SWAPs to
pack the inputs so that only $t=n/2-m$ input qubits remain on one processor.
We then apply the same synthesis procedure, which costs
$O(2^{2t})=O(2^{n-2m})$. Hence the total communication cost is
$O(\max\{2^{n-2m},\,n\})$.

\paragraph{Lower bound.}
We next prove matching lower bounds. The first captures the exponential
space-bounded lower bound; the second gives the linear lower bound
even with unlimited ancillas.

\begin{theorem}[Informal $k=2$ version of \Cref{lem:lowapprox}]\label{thm:inf-general-exp-lower}
For $\epsilon\geq 2^{-n}$, there exists an $n$-qubit unitary $U$ such that any
$\epsilon$-approximate implementation of $U$ requires
$
\Omega\!\left(\frac{2^{n-2m}\log(1/\epsilon)}{n}\right)
$
nonlocal two-qubit gates.
\end{theorem}

We begin with \emph{exact} implementation. Fix a balanced partition and a pattern
of $\ell$ nonlocal CNOT gates\footnote{
    For asymptotics, we may assume the nonlocal gates are CNOTs, since any
    two-qubit gate can be implemented by a constant number of CNOTs and
    single-qubit gates.
}. Any $(n+2m)$-qubit distributed circuit can be written as
\begin{equation} \label{eq:normal}
(A_0\otimes B_0)\cdot {\rm CNOT}_1\cdot (A_1\otimes B_1)\cdot {\rm CNOT}_2 \cdots {\rm CNOT}_\ell \cdot (A_\ell\otimes B_\ell),
\end{equation}
an alternating circuit of local unitaries $A_i,B_i$ interleaved with
nonlocal CNOT gates. Each $A_i$ or $B_i$ acts on a space of dimension
$2^{n/2+m+O(1)}$, so the whole family of such circuits have
$O(\ell\,2^{n+2m})$ real parameters. Since the full unitary group $U(2^n)$ has
real dimension $4^n$, parameter counting forces
\[
O(\ell\,2^{n+2m})\ge 4^n \Longrightarrow \ell=\Omega(2^{n-2m}).
\]

To make this lower bound robust under approximation, we count $\epsilon$-balls instead. 
The circuit family generated by \eqref{eq:normal} can be viewed as the image $\mathcal{I}$
of a polynomial map, and by a recent covering-number bound for polynomial
images \cite{zhang2023covering}, we obtain an upper bound
$\exp(O(\ell\,2^{n+2m}n))$ on the number of $\epsilon$-balls needed to cover
$\mathcal{I}$. In contrast, the number of $2\epsilon$-balls needed to cover the whole
$U(2^n)$ is known as 
$\exp(\Omega(4^n\log \frac1{\epsilon}))$ \cite{barthel2018fundamental}. If we
assume $\mathcal{I}$ approximates every unitary in $U(2^n)$ to within $\epsilon$,
then any $\epsilon$-cover of $\mathcal{I}$ must be a $2\epsilon$-cover of $U(2^n)$. 
Comparing the two bounds gives
\[
O(\ell\,2^{n+2m} n)\ge 4^n\log (1/\epsilon)
\Longrightarrow
\ell=\Omega\!\left(\frac{2^{n-2m}\log(1/\epsilon)}{n}\right).
\]

We now turn to the second ingredient, which yields the linear term.

\begin{theorem}[Informal $k=2$ version of \Cref{lem:twoomegan}]\label{thm:inf-general-lin-lower}
There exists an $n$-qubit unitary $U$ that requires $\Omega(n)$ nonlocal two-qubit gates, even with unlimited ancillas and approximations.
\end{theorem}

A natural first attempt is to use the unitary computing the inner-product
function ${\rm IP}$, i.e.,
\(
U_{\mathrm{IP}}\ket{x,y,z}=\ket{x,y,\,z\oplus \mathrm{IP}(x,y)},
\)
since ${\rm IP}$ is known to require linear quantum communication
complexity. However, $U_{\rm IP}$ only requires $O(1)$ communication 
in our model: because one can place relevant pairs
$(x_j,y_j)$ on the same processor and compute most of the inner products
locally, leaving only $O(1)$ communication. To overcome this issue,
we instead consider the \emph{shifted} inner-product function
\[
\mathrm{SIP}(i,x,y)=\bigoplus_{j=0}^{n-1} x_j\,y_{(j+i)\bmod n},
\]
where the shift $i$ is part of the input. Since the partition must be fixed
before the input is known, for any balanced partition there exists a choice of
$i$ such that $\Omega(n)$ of the relevant pairs are separated across the two
processors. This forces the corresponding unitary $U_{\mathrm{SIP}}$ to 
require $\Omega(n)$ communication under any balanced partition.

Combining \Cref{thm:inf-general-exp-lower,thm:inf-general-lin-lower}, and then
setting $\epsilon=2^{-n}$ in the first theorem, we obtain the exact lower bound
$\Omega(\max\{2^{n-2m},\,n\})$. Hence 
\Cref{thm:inf-general-upper} is asymptotically tight.

\subsubsection{Quantum Fourier Transform (QFT)}

We next study the communication complexity of the $n$-qubit QFT unitary, denoted by 
$\QFT_n$, defined for simplicity without the final bit reversal\footnote{The
same upper and lower bounds hold for the standard QFT with bit reversal; see
\Cref{sec:reversal}.} (see \Cref{fig:qft-overview}). 

\begin{theorem}[Informal version of \Cref{thm:qft}]\label{thm:inf-qft}
Exact implementation of $\QFT_n$ requires $\Theta(n)$ nonlocal two-qubit gates.
\end{theorem}

\paragraph{Upper bound.} The upper bound uses the standard grouping of the
$\QFT_n$ circuit into layers $S_1,\ldots,S_{n-1}$ by target qubit
\cite{neumann2020imperfect}, and requires only one ancilla. Under the partition
where Processor~$A$ holds the first $n/2$ qubits and Processor~$B$ holds the
last $n/2$ qubits, each layer $S_i$ can be implemented with two nonlocal
SWAPs: move the target qubit to an ancilla on the other processor, perform the
remaining gates locally, and swap it back. This gives an exact implementation
using $O(n)$ communication.

\begin{figure}[htbp]
    \centering
    \begin{tikzpicture}
    \yquantdefinebox{dots2}[inner sep=0pt]{$\dots$}
    \yquantdefinebox{vdotsqft2}[inner sep=0pt]{$\vdots$}    
    \yquantdefinebox{ghostqft2}[inner sep=0pt]{$\quad\quad$}
    \begin{yquant}
        qubit {$\ket{j_1}$} j1;
        qubit {$\ket{j_2}$} j2;
        nobit ellipsis;
        vdotsqft2 ellipsis;
        qubit {$\ket{j_{n-1}}$} jn1;
        qubit {$\ket{j_{n}}$} jn;
        h j1;
        [this subcircuit box style={dashed, "$S_1$"}]
        subcircuit {
            qubit {} j1;
            qubit {} j2;
            nobit ellipsis;
            qubit {} jn1;
            qubit {} jn;
            
            box {$R_2$} j1 | j2;
            ghostqft2 ellipsis;
            dots2 ellipsis;
            box {$R_{n-1}$} j1 | jn1;
            box {$R_{n}$} j1 | jn;
        } (j1, j2, ellipsis, jn1, jn);
        h j2;
        [this subcircuit box style={dashed, "$S_2$"}]
        subcircuit {
            qubit {} j1;
            qubit {} j2;
            nobit ellipsis;
            qubit {} jn1;
            qubit {} jn;
            dots2 ellipsis;
            box {$R_{n-2}$} j2 | jn1;
            box {$R_{n-1}$} j2 | jn;
        } (j1, j2, ellipsis, jn1, jn);
        dots2 ellipsis;
        h jn1;
        [this subcircuit box style={dashed, "$S_{n-1}$"}]
        subcircuit {
            qubit {} j1;
            qubit {} j2;
            nobit ellipsis;
            qubit {} jn1;
            qubit {} jn;
            box {$R_{2}$} jn1 | jn;
        }  (j1, j2, ellipsis, jn1, jn);
        h jn;
    \end{yquant}
\end{tikzpicture}
    \caption{The QFT circuit without bit reversal, and the groupings of controlled rotation gates.}
    \label{fig:qft-overview}
\end{figure}

\paragraph{Lower bound.}
The lower bound is more subtle and continues to hold even with unlimited
ancillas. A natural first idea is to use $\QFT_{2n}$ to transmit an $n$-bit
string $x$ from Alice to Bob, and then invoke the fact that sending $n$ bits
requires $\Omega(n)$ communication. For instance, using the standard QFT with
bit reversal, $\QFT_{2n}^{\uparrow}$, one has
$
\QFT_{2n}^{\uparrow}\ket{x}_A\ket{0^n}_B
=
H^{\otimes n}\ket{0^n}_A \otimes \QFT_n^{\uparrow}\ket{x}_B,
$
so Bob could apply $(\QFT_n^{\uparrow})^{-1}$ to recover $\ket{x}$. 

However, this approach fails for three reasons:
(i) The lower bound must hold for \emph{every} balanced partition, whereas
    the above protocol depends on a particular partition.
(ii) Our target unitary is $\QFT_{2n}$ \emph{without} the final bit
    reversal, so even for a fixed partition it is not clear how to use it to
    transmit $x$ from Alice to Bob.
(iii) Information-theoretic transmission arguments are robust under
    approximation, but \Cref{thm:inf-qft-approx} shows that approximate QFT can
    be implemented with only $O(\log n)$ communication. Hence such an argument
    cannot prove the $\Omega(n)$ lower bound for exact implementation.

Instead, we use rank arguments from communication complexity. We sketch the
proof only for the specific partition in which Alice and Bob
hold the first $n$ and the last $n$ qubits, respectively, 
and extending it to any partitions is straightforward.
Consider the following experiment:
\begin{enumerate}
    \item Alice and Bob prepare the state $\ket{x}_{A}\ket{0^n}_{B}$, where
    $x\in\{0,1\}^n$ is uniformly random.
    \item They apply $\QFT_{2n}$.
    \item Bob applies $H^{\otimes n}$ and measures his qubits to obtain
    $y\in\{0,1\}^n$.
\end{enumerate}
This produces a joint output probability matrix $M_{x,y}$. 
By standard rank arguments \cite{zhang2012quantum}, exactly generating $M$ requires
$\Omega(\log {\rm rank}(M))$ communication. Thus it suffices to show that $M$
has exponentially large rank. We do this by showing that, up to an invertible
transformation, the $y$-th column of $M$ is the evaluation of a degree-$y$
polynomial $f_y$ on a set $(a_x)_{x\in[n]}$ containing at least $2^{n-1}$
distinct values. It follows that the first $2^{n-1}$ columns of $M$ are
linearly independent, since a nonzero polynomial of degree $\leq 2^{n-1}$
has at most $2^{n-1}$ roots. Hence ${\rm rank}(M)\ge 2^{n-1}$,
and implementing $\QFT_{2n}$ requires
$\log 2^{n-1}=\Omega(n)$ communication.

\smallskip \emph{Approximation.} Interestingly, this exact lower bound is \emph{not} robust under approximation.
\begin{theorem}[Informal version of \Cref{lem:aqft-log}] \label{thm:inf-qft-approx}
$\QFT_n$ can be implemented with $O(\log n)$ nonlocal two-qubit gates,
up to $1/\poly(n)$ approximation error.
\end{theorem}

This follows from the approximate quantum Fourier transform (AQFT),
obtained by dropping all controlled rotations whose angles are smaller than
$2\pi/2^{\Omega(\log n)}$, while still approximating $\QFT_n$ within
$1/\poly(n)$ error. Under the same balanced partition, each group $S_i$ spans only $O(\log n)$ qubits, so only $O(\log n)$ groups remain
nonlocal. Applying the previous exact construction to these groups yields an
$O(\log n)$-communication implementation.

\subsubsection{Clifford Circuits}

Finally, we analyze the communication
complexity of $n$-qubit Clifford circuits. Although Clifford circuits 
require $\widetilde{\Theta}(n^2)$ gates in a non-distributed
implementation, we show that 
\begin{theorem}[Informal version of \Cref{thm:cnot}]
    $n$-qubit Clifford circuit can be 
    implemented with $\Theta(n)$ non-local two-qubit gates.
\end{theorem}

To prove this theorem, it suffice to prove the same bound for CNOT circuits, by
the canonical decomposition of Clifford circuits \cite{aaronson2004improved}.

\paragraph{Upper bound.} 
An $n$-qubit CNOT circuit is specified by an invertible matrix
$M\in\mathbb{F}_2^{n\times n}$ acting as $\ket{x}\mapsto\ket{Mx}$. We first show
that if $M$ is lower triangular, then the circuit can be implemented with
$O(n)$ communication. Indeed, $M$ admits the block decomposition
\[
\begin{pmatrix} M_{AA} & 0\\ M_{BA} & M_{BB}\end{pmatrix}
=
\begin{pmatrix} M_{AA} & 0\\ 0 & \mathbbm{I}\end{pmatrix}
\begin{pmatrix} \mathbbm{I} & 0\\ M_{BA} & \mathbbm{I}\end{pmatrix}
\begin{pmatrix} \mathbbm{I} & 0\\ 0 & M_{BB}\end{pmatrix},
\]
where the first and third factors are local. The middle factor can be
implemented row by row: for each target qubit on Processor~$B$,
Processor~$A$ computes the relevant parity onto one ancilla, sends it through a
single nonlocal CNOT, and then uncomputes the ancilla. This gives an
implementation using $O(n)$ nonlocal CNOTs.

For a general CNOT circuit, we use a PLU decomposition $M=PLU$, where $P$ is a
permutation matrix, $L$ is lower triangular, and $U$ is upper triangular. The
permutation $P$ can be implemented using $O(n)$ nonlocal SWAPs, while $L$ and
$U$ each require $O(n)$ nonlocal CNOTs by the same argument. Therefore, any
$n$-qubit CNOT circuit can be implemented using $O(n)$ nonlocal gates.

\paragraph{Lower bound.} A fixed-partition lower bound is immediate: the CNOT
circuit $\ket{x,y}\mapsto\ket{x,x\oplus y}$ lets Processor~$A$ send the entire
$n$-bit string $x$ to Processor~$B$, and hence requires $\Omega(n)$
communication. But this does not suffice in our model, since the partition may
be chosen to make the circuit mostly local.

Instead, we construct a CNOT circuit that is hard for \emph{every} balanced
partition. By the probabilistic method, for all sufficiently large $n$ there
exists an invertible matrix $M\in\mathbb{F}_2^{n\times n}$ such that every
$n/2\times n/2$ submatrix of $M$ has rank $\Omega(n)$. Let $T$ be the
corresponding CNOT circuit. We show that under any balanced partition, $T$
enables the transmission of $\Omega(n)$ bits from one processor to the other,
and therefore requires $\Omega(n)$ communication. This lower bound remains
valid even with unlimited ancillas and approximation.
\section{Preliminaries}\label{sec:prelim}

We assume basic familiarity with quantum computing and quantum information; for
a comprehensive introduction see \cite{nielsen2010quantum}. In this section, we review some
backgrounds used extensively in this work.

\subsection{Norms and Covering Number}

We first define matrix norms that will be used in this paper.


\begin{definition}[Frobenius norm]
    For a matrix $A\in \mathbb{C}^{d\otimes d}$, 
    the Frobenius norm is defined as
    \[
    \|A\|_F:=\sqrt{\sum_{i,j=1}^d |A_{i,j}|^2}.
    \]
\end{definition}

\begin{definition}[Spectral norm]
    For a matrix $A\in \mathbb{C}^{d\otimes d}$, the spectral norm is defined as
    \[
    \|A\|_2\;:=\;\max_{\|x\|=1}\|Ax\|,
    \]
    where $\|\cdot\|$ denotes the Euclidean norm of a vector.
\end{definition}

\begin{fact}
    For a matrix $A\in \mathbb{C}^{d\otimes d}$, 
    \(
    \|A\|_2\leq \|A\|_F \leq \sqrt{d}\|A\|_2.
    \)
\end{fact}

Next, we introduce several basic concepts in metric geometry.

\begin{definition}[$\epsilon$-neighborhood]
Let $(M,\mathrm{dist})$ be a metric space and $S\subseteq M$. For $\epsilon>0$, the 
$\epsilon$-neighborhood of $S$ is
\[
S_{\epsilon}\ :=\ 
\bigcup_{x\in S} B(x,\epsilon),
\]
where $B(x,\epsilon):=\{y\in M:\mathrm{dist}(x,y)\le \epsilon\}$ is the $\epsilon$-ball
around $x$. 
\end{definition}

\begin{definition}[$\epsilon$-cover]
Given $S,T\subseteq M$ and $\epsilon>0$, the set $T$ is an $\epsilon$-cover of $S$ if $S\subseteq T_{\epsilon}$.
\end{definition}

\begin{definition}[$\epsilon$-covering number]
Let $(M,\mathrm{dist})$ be a metric space, $S\subseteq M$, and $\epsilon>0$. The
$\epsilon$-covering number of $S$ under ${\rm dist}$ is
\[
\mathcal{N}(S,\mathrm{dist},\epsilon)
\ :=\ \min\Bigl\{N\in\mathbb{N}:\ \exists\,x_1,\ldots,x_N\in S\ \text{such that}\ 
S\subseteq \bigcup_{i=1}^N B(x_i,\epsilon)\Bigr\}.
\]
\end{definition}

The following two lemmas bounds the covering number of the unitary group and the
image of a polynomial map respectively, which will be used to prove the
space-bounded lower bound of unitary communication complexity.
\begin{lemma}[Lemma 1 of \cite{barthel2018fundamental}] \label{lem:suupper}
    For $0<\epsilon\leq \frac{1}{10}$, the $\epsilon$-covering number of unitary group $U(d)$ under spectral norm satisfies
    \[
    \left(\frac{3}{4\epsilon}\right)^{d^2}\leq \mathcal{N}(U(d),\|\cdot\|_2,\epsilon)\leq \left(\frac{7}{\epsilon}\right)^{d^2}.
    \]
\end{lemma}

\begin{lemma}[Theorem 2.6 of \cite{zhang2023covering}] \label{lem:cnupper}
    Let $\mathcal{K}: [-1,1]^n\to \mathbb{R}^{N}$ be a polynomial map of
    degree $K$, and $V$ be the image of $\mathcal{K}$. Then for any $\epsilon>0$,
    the $\epsilon$-covering number of V under Euclidean norm satisfies
    \[
    \log \mathcal{N}(V, \|\cdot\|, \epsilon)
    \leq n\log(1/\epsilon)+O\left(n\log N + \log K\right). 
    \]
\end{lemma}

\subsection{Quantum Circuits}  \label{sec:circuit}

\paragraph{General circuit synthesis.} We introduce a decomposition method for
synthesizing general unitaries, called the \emph{quantum Shannon decomposition}.
We first define a special kind of quantum gate, called \emph{uniformly-controlled rotation (UCR)} gates.

\begin{definition}[Uniformly controlled rotation (UCR)]
An $n$-qubit UCR is a block-diagonal unitary
\(
\operatorname{diag}\!\left(
R_P(\theta_0),R_P(\theta_1),\ldots,R_P(\theta_{2^{n-1}-1})
\right),
\)
where $P\in\{X,Y,Z\}$ and
\(
R_P(\theta)=e^{-i\theta P/2}
\)
is the single-qubit rotation around the $P$ axis.

Equivalently, as illustrated in Figure~\ref{fig:ucr}, the first $n-1$ qubits act as control qubits, and for each computational-basis string on these controls, 
rotation $R_P(\theta_j)$ is applied to the target qubit $q_n$.
\begin{figure}[htbp]
\[
\Qcircuit @C=1em @R=.7em {
    \lstick{q_1}    & \qw & \controlu\qw      & \qw &   & & \controlo\qw                & \control\qw                & \controlo\qw                & \qw & \push{\cdots\ \ \ } & \control\qw                               & \qw \\
    \lstick{q_2}    & \qw & \controlu\qw\qwx  & \qw &   & & \controlo\qw\qwx            & \controlo\qw\qwx           & \control\qw\qwx             & \qw & \push{\cdots\ \ \ } & \control\qw\qwx                           & \qw \\
    \lstick{\vdots} & {/}\qw & \controlu\qw\qwx  & \qw &:= & & \controlo\qw\qwx            & \controlo\qw\qwx           & \controlo\qw\qwx            & \qw & \push{\cdots\ \ \ } & \control\qw\qwx                           & \qw \\
    \lstick{q_n}    & \qw & \gate{R_P}\qw\qwx & \qw &   & & \gate{R_P(\theta_0)}\qw\qwx & \gate{R_P(\theta_1)}\qw\qwx & \gate{R_P(\theta_2)}\qw\qwx & \qw & \push{\cdots\ \ \ } & \gate{R_P(\theta_{2^{n-1}-1})}\qw\qwx & \qw
}
\]
\caption{An $n$-qubit uniformly controlled rotation. The first $n-1$ qubits select which rotation $R_P(\theta_j)$ is applied to the target qubit $q_n$.}
\label{fig:ucr}
\end{figure}
\end{definition}

\begin{lemma}[Theorem 13 of \cite{shende2005synthesis}] \label{lem:ucd}
    For $P\in\{X,Y,Z\}$, an $n$-qubit UCR can be decomposed into
    two $(n-1)$-qubit UCRs and two CNOT gates, as shown in \Cref{fig:ucd}.
    \begin{figure}[htbp!] 
    \[
    \Qcircuit @C=1em @R=.7em { %
       \lstick{q_1} & \qw & \controlu\qw      & \qw &   & & \qw               & \ctrl{3}\qw      & \qw               & \ctrl{3}\qw & \qw  & && \ctrl{3}\qw      & \qw               & \ctrl{3}\qw  & \qw & \qw\\ 
       \lstick{q_2} & \qw & \controlu\qw\qwx  & \qw &   & & \controlu\qw      & \qw              & \controlu\qw      & \qw         & \qw  & && \qw              & \controlu\qw      & \qw          & \controlu\qw &\qw\\
       \lstick{\vdots} & {/}\qw & \controlu\qw\qwx  & \qw & = & & \controlu\qw\qwx  & \qw              & \controlu\qw\qwx  & \qw         & \qw  &=&& \qw              & \controlu\qw\qwx  & \qw          & \controlu\qw\qwx &\qw\\ 
       \lstick{q_n} & \qw & \gate{R_P}\qw\qwx & \qw &   & & \gate{R_P}\qw\qwx & \targ            & \gate{R_P}\qw\qwx & \targ       & \qw  & && \targ            & \gate{R_P}\qw\qwx & \targ        & \gate{R_P}\qw\qwx &\qw\\ 
    }
    \]

    \caption{Decomposition of UCR gates. The leftmost is an $n$-qubit UCR,
    the middle is its decomposition into two $(n-1)$-qubit UCRs and two CNOT
    gates, and the rightmost is an alternative decomposition.}\label{fig:ucd}
\end{figure}
\end{lemma}

\begin{lemma}[Quantum Shannon decomposition~\cite{shende2005synthesis}]\label{lem:qsd}
Given an $n$-qubit unitary $U$, it can be decomposed 
as four $(n-1)$-qubit unitaries and three UCRs, 
as shown in \Cref{fig:qsd}.
\begin{figure}[htbp!]
\[
\begin{array}{ccc}
\Qcircuit @C=1em @R=0.7em {
   \lstick{q_1} & \qw     &\multigate{1}{U} & \qw &  \raisebox{-2.5em}{=} & & \qw & \qw      & \gate{R_Z} & \qw         & \gate{R_Y} & \qw       & \gate{R_Z} & \qw  & \qw \\
   \lstick{\vdots} & {/} \qw &\ghost{U}    & \qw &          & & {/} \qw & \gate{V_1} & \gate{} \qwx  & \gate{V_2}  & \gate{}\qwx & \gate{V_3} & \gate{}\qwx & \gate{V_4} & \qw
}
\end{array}
\]
\caption{
Quantum Shannon decomposition of an $n$-qubit unitary $U$. 
The slashed wire denotes the group of last $n-1$ qubits.
The leftmost is the
original unitary $U$, and the rightmost is its decomposition into four
$(n-1)$-qubit unitaries $V_1,\ldots,V_4$ and three UCRs.
} \label{fig:qsd}
\end{figure}
\end{lemma}

\paragraph{Special quantum circuits.} \emph{CNOT circuits} refer to circuits
generated by only CNOT gates. A CNOT maps 2-qubit state $\ket{x, y}$ to
$\ket{x,\, x\oplus y}$, which can be written as an invertible matrix in
$\mathbb{F}_2^{2\times 2}$: 
\[\begin{pmatrix}
    1 & 0 \\
    1 & 1 \\
\end{pmatrix} 
\begin{pmatrix}
    x \\
    y \\
\end{pmatrix}
=
\begin{pmatrix}
    x \\
    x\oplus y \\
\end{pmatrix}.
\]
It follows readily that an $n$-qubit CNOT circuit $T$ acts as a reversible
linear transformation on the $n$ input bits, and hence can be represented by an
invertible matrix $M\in\mathbb{F}_2^{n\times n}$ such that $T\ket{x}=\ket{Mx}$ for
all $x\in\{0,1\}^n$ \cite{moore2001parallel}.

\emph{Clifford circuits} refer to circuits generated by the basic gate set
$\{H,S,{\rm CNOT}\}$. The following lemma gives a canonical form of Clifford
circuits.

\begin{lemma}[\cite{aaronson2004improved}] \label{lem:clifford}
    Any Clifford circuit can be implemented by an 11-layer sequence
    H-C-S-C-S-C-H-S-C-S-C, where H denotes a layer of $H$ gates,
    S denotes a layer of $S$ gates, and C denotes a CNOT circuit.
\end{lemma}

\subsection{Quantum Information Theory}

\begin{definition}[Von Neumann entropy]
    Let $\rho$ denote the quantum state of a system $A$, and
    $\rho = \sum_i \eta_i \ket{\phi_i}\bra{\phi_i}$ denote a spectral
    decomposition of $\rho$, where $\{\ket{\phi_i}\}$ is an eigenbasis for $\rho$. The Von Neumann entropy of
    $\rho$, denoted by $S(A)_\rho$ or $S(\rho)$, is defined as
    \[
        S(A)_\rho = S(\rho) = -\sum_i \eta_i \log(\eta_i).
    \]
\end{definition}

\begin{lemma}[Fannes-Audenaert inequality \cite{audenaert2007sharp}] \label{lem:fannes}
    For two quantum states $\rho$ and $\sigma$ of dimension $d$, 
    \[
    |S(\rho)-S(\sigma)|\leq D \log (d-1) + H(D, 1-D),
    \]
    where $D$ is the trace distance between $\rho$ and $\sigma$, and $H$ denotes the binary Shannon entropy.
\end{lemma}

\begin{definition}[Mutual information]
    Given a quantum state $\rho$ that describes the joint systems $A$ and $B$,
    the mutual information between $A$ and $B$ is given by
    \[
        I(A:B)_{\rho} = S(A)_{\rho} + S(B)_{\rho} - S(AB)_{\rho}.
    \]
\end{definition}

We often omit the subscript $\rho$ when the quantum state is
clear from context. For example, we will write $I(A:B)$ instead of $I(A:B)_\rho$.

\begin{fact} \label{lem:info}
    A two-qubit quantum gate across two systems $A$ and ${B}$ increase $I({A}:{B})$ by at most $4$.
\end{fact}

\begin{proof}
    A two-qubit gate can be implemented by exchanging $2$ qubits and local computation, 
    where each qubit exchanged increase $I(A:B)$ by at most $2$, 
    and local computation does not increase $I(A:B)$ \cite{nielsen2010quantum}.
\end{proof}

Finally, we state a well-known lower bound on the communication complexity of the inner-product function.

\begin{lemma}[\cite{cleve1998quantum}]\label{lem:ip}
Suppose Alice and Bob holds $n$-bit inputs $x$ and $y$ respectively.
For any constant $\epsilon > 0$, computing
\[\mathrm{IP}(x,y) :=\bigoplus_{i=0}^{n-1} x_i y_i\]
with probability $1/2+\epsilon$ 
requires exchanging $\Omega(n)$ qubits.
\end{lemma}

\section{Distributed Communication Complexity and Its Variants}\label{sec:formaldefs}

In this section, we present the formal definitions of distributed communication
complexity, along with its multi-party and approximate variants. Throughout the
paper, we measure the communication cost of a distributed quantum circuit in
terms of the number of \emph{nonlocal two-qubit gates}, i.e., two-qubit gates
acting on qubits residing on different processors. This is without loss of
generality, as any other communication measure, e.g., the number of qubits
exchanged, is equivalent up to a constant factor. Moreover, we place no restriction on the type of nonlocal two-qubit gates, as this choice does not affect the asymptotic analysis.

The set of qubits on which the given unitary acts is referred to as the \emph{input qubits} (or \emph{data qubits}). In our definitions, we optimize the communication cost over all \emph{balanced partitions} of the input qubits. This reflects practical scenarios in which we are allowed to arrange qubits when distributing a quantum computing task across multiple processors, and it also adds nontrivial structure to our lower bound proofs. Formally, a \emph{partition} of $n$ input qubits among $k$
processors is specified by a map \( \pi:[n]\to[k], \) where $\pi(i)$ denotes the
processor that initially holds the $i$-th input qubit. For $\ell\in[k]$, write \(
\pi^{-1}(\ell):=\{\,i\in[n]:\pi(i)=\ell\}\). We call $\pi$
balanced if \( \left|\pi^{-1}(\ell)\right| = n/k+O(1) \) for
any $\ell\in[k]$.

The standard distributed communication complexity of
a unitary $U$
is defined to be the minimum
communication required to \emph{exactly} implement $U$ on two
processors, optimized over all balanced partitions of the input qubits and all
possible circuit implementations. It also takes a parameter $m$ for the number of
clean ancilla qubits available to each processor.

\begin{definition}[Distributed communication complexity, DCC]
\label{def:commtwo}
Consider implementing an $n$-qubit unitary $U$ on two quantum processors $A$ and
$B$. The distributed communication complexity of $U$ with ancilla bound $m$ is defined as
\[
\CC_m(U)\;:=\;\min_{\text{balanced }\pi:[n]\to[2]} \CC_{m}(U\mid\pi),
\]
where $\CC_{m}(U\mid\pi)$ denotes the minimum number of nonlocal two-qubit gates required to implement $U$ when $A$ and $B$ hold the input qubits indexed by $\pi^{-1}(1)$ and $\pi^{-1}(2)$ respectively,
along with $m$ clean ancilla qubits per processor.

For simplicity, we define $\CC_{\infty}(U):=\inf_{m\geq 0}\CC_m(U).$
\end{definition}

\begin{example} 
    Below are some examples of distributed communication complexity.
    \begin{enumerate}
        \item 
        Given a Boolean function $f:\{0,1\}^n\to\{0,1\}$, define unitary
        $U_f:\ket{x,y}\mapsto\ket{x,y\oplus f(x)}$. Then $\CC_\infty(U_f)$
        equals the quantum communication complexity of $f$ optimized over all
        balanced partitions of $n$ input bits. 
        \item Given an $n$-qubit quantum circuit $T$, $\CC_0(T)$ is at most
        the number of two-qubit gates in $T$. Thus for any $n$-qubit
        unitary $U$, $\CC_0(U)=O(4^n)$ as any $n$-qubit unitary can be
        realized using $O(4^n)$ CNOT gates~\cite{shende2004minimal}.
    \end{enumerate}
\end{example}

\begin{remark}
Our definition differs from conventional quantum communication complexity in two important respects: 
(i) we study the coherent implementation of a quantum unitary, rather than the computation of a classical function value; and 
(ii) we minimize over all balanced partitions of the input qubits, as motivated by the distributed unitary-synthesis setting.
Consequently, $\CC_\infty(U_f)$ and the conventional quantum communication complexity of $f$ can differ dramatically.
For example, let
\[
f(x_1,\ldots,x_{n/2},y_1,\ldots,y_{n/2})=\bigoplus_{i=1}^{n/2} x_i y_i
\]
be the inner-product function.
Under the standard partition, where one party receives $(x_1,\ldots,x_{n/2})$ and the other receives $(y_1,\ldots,y_{n/2})$, the quantum communication complexity of $f$ is $\Omega(n)$.
In contrast, for the unitary $U_f$, we may choose a balanced partition that places each pair $(x_i,y_i)$ on the same processor.
Then each processor can compute the parity of its local products, and only $O(1)$ communication is needed to combine the two parities.
This flexibility is precisely why, in later lower-bound arguments, we use shifted variants of the inner-product function instead.
\end{remark}

We now extend the notion to settings involving more than two processors, where the processors are connected according to a specified topology.

\begin{definition}[Multi-party DCC]
\label{def:commk}
Let $G=([k],E)$ be an undirected graph describing the topology among $k$ quantum processors. 
Consider implementing an $n$-qubit unitary $U$ on these processors, where each processor holds $n/k+O(1)$ input qubits and $m$ ancilla qubits, and nonlocal two-qubit gates are permitted only between processors $(i,j)\in E$. 

The communication complexity of $U$ over $G$ with ancilla bound $m$ is defined as
\[
\CC_m^{G}(U)\;:=\;\min_{\text{balanced }\pi:[n]\to[k]}\;\CC_m^{G}(U\mid\pi),
\]
where $\CC_m^{G}(U\mid\pi)$ denotes the minimum number of nonlocal two-qubit gates required to implement $U$ when each processor $i$ initially holds the input qubits indexed by $\pi^{-1}(i)$ and $m$ clean ancilla qubits.

For the complete graph $K_k$, we define
\[
\CC_m^{(k)}(U)\;:=\;\CC_m^{K_k}(U),
\]
and refer to it as the $k$-party distributed communication complexity of $U$.
\end{definition}

\begin{fact}
    $\CC^{(k)}_m(U)$ is a non-decreasing function of $k$ for any fixed $U$ and $m$.
\end{fact}

We also define an approximate variant of the above definition, where the goal is
relaxed to implementing a unitary that approximates $U$ up to an error parameter $\epsilon$.

\begin{definition}[Approximate DCC] 
Given an $n$-qubit unitary $U$ and $m\geq 0$,
\[
    \CC_m(U;\epsilon)\;:=\;\min_{\text{unitary }V:\|V-U\|_2\leq \epsilon}\;\CC_m(V).
\]
Moreover, the notation $\CC^G_m(U;\epsilon)$ and
$\CC^{(k)}_m(U;\epsilon)$ are defined analogously for multi-party
distributed communication complexity over a graph $G$ and a complete graph $K_k$,
respectively. 
\end{definition}

\section{Asymptotically Optimal Bound for $\CC_m^{(k)}(U)$} \label{sec:commk}

We present an asymptoticaly optimal characterization for the $k$-party distributed communication
complexity of $n$-qubit unitaries, as shown in \Cref{thm:commk}. 
\begin{theorem} \label{thm:commk}
    Given integers $m\geq 0, k\geq 2$, we have that
    \begin{enumerate}
        \item[(i)] for any $n$-qubit unitary $U$,  $\CC_m^{(k)}(U)=O(\max\{4^{(1-1/k)n-m},n\})$;
        \item[(ii)] there exists $n$-qubit unitary $U$ such that $\CC_m^{(k)}(U)=\Omega(\max\{4^{(1-1/k)n-m},n\})$.
    \end{enumerate}
\end{theorem}

\begin{example}
  We present several special cases of \Cref{thm:commk} to illustrate it contents.
\begin{itemize}
  \item \textbf{No ancillas, two processors.} When $m=0, k=2$, we have
  \(\CC_0^{(2)}(U)=\Theta(2^n)\). Since implementing arbitrary $n$-qubit unitary requires
  $\Theta(4^n)$ two-qubit gates \cite{shende2004minimal}, a trivial scheme of 
  partitioning input qubits into two balanced parts will yield $O(4^n)$ nonlocal two-qubit gates.
  However, \Cref{thm:commk} shows that up to circuit transformation, the optimal
  partition only requires $\Theta(2^n)$ nonlocal two-qubit gates, achieving a
  quadratic improvement over the trivial scheme.
  \item \textbf{No ancillas, multiple processors.} 
  When $m=0$ and $k\in\{2,3,\ldots,n\}$, the bound of $\CC^{(k)}_{0}(U)$ scales from
  $\Theta(2^n)$ at $k=2$, $\Theta(4^{2n/3})$ at $k=3$, to $\Theta(4^n)$ at $k=n$. 
  In the extreme $k=n$ case, each processor holds $O(1)$ input qubits, so the number of 
  nonlocal two-qubit gates required matches the number of two-qubit gates required to 
  implement $U$ in the worst case, which is $\Theta(4^n)$.
  \item \textbf{Many ancillas.} When $m\ge n$, we have \(\CC^{(k)}_{m}(U)=\Theta(n)\).
  Here the available ancilla qubits suffice to buffer all inputs, so one can move
  all input qubits to a single processor using $O(n)$ nonlocal SWAPs and complete
  the computation locally. The lower bound shows $\Omega(n)$ communication is
  also necessary in the worst case.
\end{itemize}
\end{example}

The rest of this section is to prove \Cref{thm:commk}, which consists
of three parts. First, \Cref{sec:generalup} presents a synthesis algorithm that implements any
$n$-qubit unitary with $O(\max\{4^{(1-1/k)n-m},n\})$ nonlocal two-qubit gates,
matching the upper bound. We further generalize the algorithm to general topology settings.
Next, \Cref{sec:generallow} proves that there exists an $n$-qubit unitary
that requires $\Omega(4^{(1-1/k)n-m}\log(1/\epsilon)/n)$ nonlocal two-qubit gates to approximate within error $\epsilon$ when the
ancillas are limited to $m$ per processor, which directly implies the $\Omega(4^{(1-1/k)n-m})$ lower bound in the exact model.
Finally, \Cref{sec:generallowub} shows that $\Omega(n)$ communication is
necessary in the worst case, even when $m$ is unbounded and approximation is allowed.

\subsection{Efficient Distributed Unitary Synthesis}
\label{sec:generalup}

To prove the upper bound, we present a distributed unitary synthesis algorithm
that implements any $n$-qubit unitary $U$ on $k$ processors, where each processor
holds $n/k+O(1)$ input qubits and $m$ ancilla qubits, using
$O\bigl(\max\{4^{(1-1/k)n-m},n\}\bigr)$ nonlocal two-qubit gates.

\begin{lemma}\label{lem:upperk}
Given any $n$-qubit unitary $U$, and $m\geq 0$, we have
\(
\CC_{m}^{(k)}(U)=O\bigl(\max\{4^{(1-1/k)n - m},n\}\bigr).
\)
\end{lemma}

The proof of \Cref{lem:upperk} consists of two steps. First, \Cref{sec:no-anc} gives a
synthesis algorithm with no ancillas, using $O(4^{(1-1/k)n})$ nonlocal
two-qubit gates. Then \Cref{sec:gen-anc} extends the algorithm to the general case with ancillas, which yields the desired bound.

\subsubsection{No-ancilla Case} \label{sec:no-anc}

We first present two helper lemmas about the distributed implementation of UCRs.

\begin{lemma}[Decomposition of UCR]\label{lem:ucdecomp}
Let $R$ be an $n$-qubit UCR gate, and fix an integer $a\in[n]$.
Assume the target qubit index $t$ satisfies $t>a$. Then $R$ can be decomposed into
\begin{enumerate}
  \item $2^{a}$ CNOT gates between the first $a$ qubits (as controls) and the target qubit $t$, and
  \item $2^{a}$ UCR gates acting only on the last $n-a$ qubits.
\end{enumerate}
\end{lemma}

\begin{proof}
  When $a=1$, the lemma follows directly from \Cref{lem:ucd}, which expresses
  $R$ as $2$ CNOTs between the first qubit and target $t$, interleaved with two
  $(n-1)$-qubit UCRs. When $a=2$, by applying \Cref{lem:ucd} twice, $R$ can be decomposed as:
  \[
  \Qcircuit @C=1em @R=.7em { %
  \lstick{q_1} & \controlu \qw & \qw & & & & \qw & \qw & \qw & \qw & \ctrl{3}
   \qw & \qw & \qw & \qw & \qw & \ctrl{3} \qw & \qw \\
  \lstick{q_2} & \controlu \qw \qwx & \qw & & & & \qw & \ctrl{2} \qw & \qw
   & \ctrl{2} \qw & \qw & \ctrl{2} \qw & \qw & \ctrl{2} \qw & \qw
   & \qw & \qw \\
  \lstick{\vdots} & \controlu \qw \qwx & \qw & = & & & \controlu \qw & \qw &
   \controlu \qw & \qw & \qw & \qw & \controlu \qw & \qw &
    \controlu \qw & \qw & \qw \\
  \lstick{t} & \gate{R_P} \qw \qwx & \qw & & & & \gate{R_P} \qw \qwx & \targ &
   \gate{R_P} \qw \qwx & \targ & \targ & \targ & \gate{R_P} \qw \qwx &
   \targ & \gate{R_P} \qw \qwx & \targ & \qw 
   \gategroup{2}{10}{4}{10}{1em}{--}
   \gategroup{2}{12}{4}{12}{1em}{--}
  }
\]
where the dashed two CNOTs cancel out, resulting in $2+2=4$ CNOTs and $4$ $(n-2)$-qubit UCRs. 
When $a>2$, by recursively applying \Cref{lem:ucd} for $a$ times, the $i$-th recursion ($i\ge 2$)
contributes $2^{i-1}$ net new CNOTs after cancellations. Thus the total number
of CNOTs between the first $a$ qubits and the target $t$ is \(2+\sum_{i=2}^{a}
2^{\,i-1}=2^{a}\). Lastly, one can easily verify that the the remaining number of 
$(n-a)$-qubit UCR gates is also $2^{a}$.
\end{proof}

\begin{lemma}\label{lem:ucdist}
Given an $n$-qubit UCR gate $R$ and $k\ge 2$, we have
\(C_{0}^{(k)}(R)\le2^{(1-1/k)n+1}-2\).
\end{lemma}

\begin{proof}

Since we may choose the qubit assignment, assume the target qubit of $R$ resides
on the $k$-th processor. Apply \Cref{lem:ucdecomp} recursively $k-1$ times with
parameter $a=n/k$, each time producing the subcircuit to be executed on one processor. 

At the $i$-th recursion step ($i=1,\ldots,k-1$), there are $2^{(i-1)n/k}$ UCRs
carried forward from earlier steps, and each UCR contributes $2^{n/k}$ nonlocal
CNOTs across the $i$-th and the rest processors by \Cref{lem:ucdecomp}.
Hence the $i$-th step contributes
\(
2^{(i-1)n/k}\times 2^{n/k}=2^{in/k}
\)
nonlocal CNOTs. Summing over $i=1$ to $k-1$, the total number of nonlocal CNOTs is
\[
\sum_{i=1}^{k-1} 2^{in/k}
=\frac{2^n-2^{n/k}}{2^{n/k}-1}
\le2^{(1-1/k)n+1}-2. \qedhere
\]
\end{proof}

Next, we analyze the cost required to decompose a unitary with respect to one
processor versus the remaining processors.

\begin{lemma} \label{lem:decompU}
Given an $n$-qubit unitary $U$ and integer $r \in [n-1]$, $U$ can be decomposed into
$4^{r}$ many $(n-r)$-qubit unitaries acting on the last $n-r$ qubits, together with
\emph{(i)} $6\times 4^{r}$ two-qubit gates across the first $r$ qubits and the last $n-r$ qubits, and 
\emph{(ii)} $3\times 4^{r}$ UCRs acting on the last $n-r$ qubits.

\end{lemma}

\begin{proof} 
  Recursively applying \Cref{lem:qsd} to $U$ for $r$ times, we get
  \begin{enumerate}
    \item $4^{r}$ unitaries acting on the last $n-r$ qubits; and
    \item for each $0\le i<r$, a collection of $3\times 4^{i}$ UCR gates on last $n-i$ qubits.
  \end{enumerate}
  For each UCR gate $R$ in item (2), first use two nonlocal SWAPs at the
  beginning and end of the subcircuit to move the target qubit of $R$ to
  the last $n-r$ qubits, and then apply \Cref{lem:ucdecomp} with parameter $a=n/k-i$ to
  decompose $R$ into CNOTs and $(n-r)$-qubit UCRs. In total, item (2) produces
\begin{enumerate}
  \item[(i)] 
  \(
  \sum_{i=0}^{r-1} 3\times 4^i\times\bigl(2^{\,r-i}+2\bigr)\le 6\times 4^{r}
  \)
  two-qubit gates across the first $r$ and last $n-r$ qubits; and
  \item[(ii)]
  \(
  \sum_{i=0}^{r-1} 3\times 4^i\times 2^{r-i}\le 3\times 4^{r}
  \)
  UCRs acting on the last $n-r$ qubits. \qedhere
\end{enumerate}  
\end{proof}

\begin{corollary}\label{lem:decompone}
Given $k$ processors with pairwise communication, each holding $n/k$ input
qubits, consider implementing an $n$-qubit unitary $U$. Then $U$ can be
decomposed into 
\begin{enumerate}
  \item $6\times 4^{n/k}+6\times 2^{n}$ nonlocal two-qubit gates, and
  \item $4^{n/k}$ unitaries on the last $k-1$ processors acting on $(1-1/k)n$ qubits.
\end{enumerate}
\end{corollary}

\begin{proof}
  Applying \Cref{lem:decompU} with parameter $r=n/k$, we obtain 
  \begin{enumerate}
    \item[(i)] $6\times 4^{n/k}$ two-qubit gates between the $i$-th and the rest processors; 
    \item[(ii)] $3\times 4^{n/k}$ UCRs acting on the last $(1-1/k)n$ qubits; and
    \item[(iii)] $4^{n/k}$ unitaries acting on the last $(1-1/k)n$ qubits.
  \end{enumerate}
  Finally, use \Cref{lem:ucdist} to decompose all UCRs in item (ii) over the remaining $k-1$ processors, 
  which in total requires
\(
3\times 4^{n/k}\times 2^{(1-2/k)n+1}=6\times 2^{n}
\)
additional nonlocal two-qubit gates. \qedhere

\end{proof}

We now design a distributed unitary synthesis algorithms with no ancillas.
\begin{lemma}\label{lem:kexpzero}
For any $n$-qubit unitary $U$,
\(
\CC_{0}^{(k)}(U)=O\left(4^{(1-1/k)n}\right).
\)
\end{lemma}

\begin{proof}
    Apply \Cref{lem:decompone} recursively $k-1$ times to distribute $U$
    over $k$ processors. 
    At the beginning of the $i$-th recursion ($i\in[k-1]$), there are
    $4^{(i-1)n/k}$ unitaries, each acting on $(1-(i-1)/k)n$ qubits, distributed
    over the remaining $k-i+1$ processors. Since each remaining processor still
    holds $n/k$ input qubits, we apply \Cref{lem:decompone} to each such
    unitary, with the parameter 
    $n$ set as $(1-(i-1)/k)n$ and $k$ set as $k-i+1$. Thus
    each application contributes $6\times 4^{n/k}+6\times 2^{(1-(i-1)/k)n}$
    nonlocal two-qubit gates and produces $4^{n/k}$ smaller unitaries on the
    last $k-i$ processors. Thus the $i$-th recursion step contributes
\[
4^{(i-1)n/k}\times\Bigl(
6\times 4^{n/k} +
6\times 2^{(1-(i-1)/k)n}
\Bigr)
= 6\times 4^{in/k} + 6\times 2^{(1+(i-1)/k)n}
\]
nonlocal two-qubit gates. Summing these contributions over $i\in[k-1]$ gives
\begin{align*}
6 \sum_{i=1}^{k-1} 4^{in/k}+6 \sum_{i=0}^{k-2} 2^{(1+i/k)n}
&= 6 \times 4^{(1-1/k)n+O(1)} + 6\times 2^n \times 2^{(1-2/k)n+O(1)}
= O\left(4^{(1-1/k)n}\right).\ \qedhere
\end{align*}

\end{proof}

\subsubsection{General Case} \label{sec:gen-anc}

Finally, we extend the previous lemma to the setting with ancilla
qubits, obtaining \Cref{lem:upperk}.
\begin{proof}[Proof of \Cref{lem:upperk}]
Write $N:=n/k+m$ for the number of qubits (inputs $+$ ancillas) available to a
single processor.\footnote{For simplicity we assume $k\mid n$, as dropping this
assumption affects only constant factors.} We first
gather all $n$ input qubits into
\(
K:=\left\lceil \frac{n}{N}\right\rceil
\)
processors,
which consumes $O(n)$ nonlocal SWAP gates. If $K=1$ the unitary can be implemented locally and we are done, so
assume $K\ge 2$. After gathering, $K-1$ processors each hold exactly $N$ input
qubits and the remaining processor holds $R$ inputs, which satisfies
\begin{equation}\label{eq:NR}
n = (K-1)N\,+\,R,\qquad 1\le R\le N .
\end{equation}

Apply Lemma~\ref{lem:decompU} with parameter $r=R$ to decompose $U$ into:
\begin{enumerate}\itemsep0pt
\item[(i)] $O(4^{R})$ two-qubit gates across the first $K-1$ and the $K$-th processors;
\item[(ii)] $O(4^{R})$ UCRs acting on the $(n-R)$ input qubits located in the first $K-1$ processors;
\item[(iii)] $4^{R}$ unitaries acting on the $(n-R)$ input qubits located in the first $K-1$ processors.
\end{enumerate}

We bound the number of nonlocal two-qubit gates needed to realize each group.

\noindent\emph{Item (i).}
Using $R\le n-N$ (which follows from \eqref{eq:NR} since $n-N=(K-2)N+R$ and $K\ge2$), we have
\(O(4^{R})\leq O\left(4^{n-N}\right)\) nonlocal two-qubit gates.

\noindent\emph{Item (ii).}
By \Cref{lem:ucdecomp}, each $(n-R)$-qubit UCR on the first $K-1$ processors can be implemented using
$O\left(2^{(1-\frac{1}{K-1})(n-R)}\right)$ nonlocal two-qubit gates. As
there are $O(4^{R})$ such UCRs, the total here is
\[
O\left(4^{R}\cdot 2^{(1-\frac{1}{K-1})(n-R)}\right)
= O\left(2^{2R + (1-\frac{1}{K-1})(n-R)}\right)
= O\left(2^{n + R - \frac{n-R}{K-1}}\right).
\]
By \eqref{eq:NR} we have $\frac{n-R}{K-1}=N$, then this becomes
$O\left(2^{n + R - N}\right)\le O\left(2^{n + (n-N) - N}\right)
=O\left(4^{n-N}\right)$, as $R\le n-N$.

\noindent\emph{Item (iii).}
By \Cref{lem:kexpzero}, each $(n-R)$-qubit unitary on the first $K-1$
processors can be implemented using
$O\left(4^{(1-\frac{1}{K-1})(n-R)}\right)$ nonlocal two-qubit gates. There
are $4^{R}$ such unitaries, so the total is
\[
4^{R}\cdot O\left(4^{(1-\frac{1}{K-1})(n-R)}\right)
= O\left(4^{n - \frac{n-R}{K-1}}\right)
= O\left(4^{n - N}\right),
\]
where we again used $\frac{n-R}{K-1}=N$.

Combining the three groups and the initial $O(n)$ SWAP gates, the overall number
of nonlocal two–qubit gates is
\(
O\left(n + 4^{\,n-N}\right) = O\left(\max\{4^{n-N},n\}\right).
\)
Recalling $N=n/k+m$ yields
\(
\CC_m^{(k)}(U)=O\left(\max\{4^{(1-1/k)n - m},n\}\right),
\)
as claimed.
\end{proof}

\paragraph{Extending to general topology.} The definition of $\CC_m^{(k)}(U)$ assumes pairwise communication between
processors, i.e., the interconnect topology is a complete graph. In practice, a
quantum network may be a line, grid, tree, etc. The following corollary shows
that the processor topology does not substantially affect the communication cost
of implementing a unitary. Concretely, the exponential term $O\left(4^{(1-1/k)n-m}\right)$
is topology-independent, while the linear term increases from $O(n)$ to $O(Dn)$,
where $D$ is the diameter of the topology graph.

\begin{corollary} \label{coro:commtopo}
    Given $n$-qubit unitary $U$, integer $m\geq 0$ and undirected graph
    $G:=([k],E)$, we have \[\CC_m^{G}(U)=O(\max\{4^{(1-1/k)n-m},D n\})\] where
    $D$ is the diameter of $G$.
\end{corollary}

The topology-induced overhead has two sources. 
\begin{enumerate}
    \item In the decomposition of
UCRs used in the proof of \Cref{lem:upperk}, the target qubit must be moved among
processors so that \Cref{lem:ucdecomp} can distribute the gate correctly. On a
clique topology, moving a target costs $1$ hop; on other topologies it may cost
up to $D\le k-1$ hops. A careful analysis shows that this extra routing
contributes at most $O(4^{(1-3/(2k))n})$ additional nonlocal two-qubit gates, and
thus does not change the exponential term. 
    \item In promoting the no-ancilla bound of \Cref{lem:kexpzero} to
    the arbitrary-ancilla setting of \Cref{lem:upperk}, we use nonlocal SWAPs to
    gather the $n$ input qubits from $k$ processors onto $k'$ processors.
    On a clique the cost per moved qubit is $1$, whereas on a general topology
    it is $D$, so the $O(n)$ term becomes $O(Dn)$ overall. 
\end{enumerate}
Therefore, the upper bound in \Cref{thm:commk} extends to any undirected graph
$G$ with diameter $D$, as stated in \Cref{coro:commtopo}. The detailed 
proof is deferred to \Cref{app:commtopo}.

\subsection{Space-Bounded Lower Bound} \label{sec:generallow}

To prove the exponential lower bound in \Cref{thm:commk}, we actually prove a 
stronger result regarding the approximate communication complexity.

\begin{lemma} \label{lem:lowapprox}
    Given integers $n>0, m\geq 0$ and $2\le k\le n$,  
    there exists an $n$-qubit unitary $U$ such that
    \[\CC_m^{(k)}(U;\epsilon)=\Omega\left(\frac{4^{(1-1/k)n-m} \log(1/\epsilon)}{n}\right),\]
    for parameter $\epsilon\in [2^{-n}, 1/20]$.
\end{lemma}

A direct corollaries follow from the above lemma by setting $\epsilon$ to be
$2^{-n}$. As an exactly implementation of $U$
trivially approximates $U$ within $2^{-n}$ error, \Cref{cor:experror} directly
implies exponential lower bound in \Cref{thm:commk}.

\begin{corollary} \label{cor:experror}
    Given integers $n>0, m\geq 0$ and $2\le k\le n$, there exists an $n$-qubit unitary $U$ such that
    \(\CC_m^{(k)}(U)\geq \CC_m^{(k)}(U;2^{-n})=\Omega\left(
    4^{(1-1/k)n-m}\right).\)
\end{corollary}

Below, we prove Lemma~\ref{lem:lowapprox} via a covering–number argument. 
First, we derive an upper bound on the
covering number of the family of unitaries realizable by circuits
that use at most $\ell$ nonlocal two–qubit gates.
Next, we obtain a lower bound on the covering number of the
unitary group $U(2^n)$. Finally, comparing these two estimates
yields the claimed lower bound on $\ell$.

\begin{proof}[Proof of Lemma~\ref{lem:lowapprox}]
  Prove by contradiction. 
Assume any $n$-qubit unitary can be $\epsilon$-approximated in spectral norm by a circuit that uses $k$ processors, where each processor initially holds $n/k+O(1)$ input qubits and $m$ ancilla qubits, and the circuit uses at most $\ell$ nonlocal two-qubit gates (across processors). For a inital balanced partition $\pi:[n]\to[k]$ and an ordered list of processor pairs $P=((p_1,q_1),\dots,(p_\ell,q_\ell))$ indicating where the nonlocal gates act, any such circuit can, without loss of generality, be written in the alternating normal form
\begin{equation}\label{eq:normal-form}
U = U_0\cdot\mathrm{CNOT}_{p_1,q_1}\cdot U_1\,\cdots\,\mathrm{CNOT}_{p_\ell,q_\ell}\cdot U_\ell,
\end{equation}
where each $U_i$ is a tensor product of local unitaries on processors (for $i>0$ the support is contained in processors $p_i$ and $q_i$), and each $\mathrm{CNOT}_{p_i,q_i}$ is a nonlocal CNOT gate acting on the first qubits of processors $p_i,q_i$. Let
\[
R := 2^{\,n/k+m+O(1)}
\]
be the local Hilbert space dimension per processor (inputs $+$ ancillas $+$ a fixed $O(1)$ overhead). Then
$U_0=\bigotimes_{j=1}^k U_0^{(j)}$ with $U_0^{(j)}\in U(R)\subseteq \mathbb{C}^{R\times R}\cong \mathbb{R}^{2R^2}$ contributes $k\cdot
2R^2$ real parameters, and each $U_i$ with $i>0$ (supported on $p_i,q_i$)
contributes $2\cdot 2R^2$ real parameters. Thus $\bigl((k+2\ell)\cdot
2R^2\bigr)$ real parameters suffice to specify \eqref{eq:normal-form}. 
Note that all those real parameters has absolute value at most 1, as they are the
real/imaginary parts of entries of unitary matrices. Let $\mathcal{K}:
[-1,1]^{(k+2\ell)2R^2} \mapsto \mathbb{C}^{2^n\times 2^n}\cong
\mathbb{R}^{2^{2n+1}}$ be the evaluation map from those real parameters to the global
unitary. Since tensoring and matrix multiplication are multilinear in the
entries, $\mathcal{K}$ is a polynomial map of degree $k+2\ell$. 


For fixed $(\pi,P)$, let $S(\pi,P)$ denote the set of unitaries realizable by
\eqref{eq:normal-form}. As $S(\pi,P)$ is contained in $\mathcal{K}$'s image,
by \Cref{lem:cnupper}, there
exist absolute constants $c_1,c_2>0$ such that
\begin{align}
   \notag \log \mathcal{N}\bigl(S(\pi,P),\|\cdot\|_F,\epsilon\bigr)
&\le (k+2\ell)R^2\log(1/\epsilon)+c_1 (k+2\ell)R^2\,n+c_2\log(k+2\ell) \\
\label{eq:here}&= (k+2\ell)R^2(\log(1/\epsilon)+c_1n)+c_2\log(k+2\ell)
\end{align}
There are at most $k^n$ choices of $\pi$ and at most $k^{2\ell}$ choices of $P$; hence by subadditivity of covering number under unions,
\begin{align}
\log \mathcal{N}\left(\bigcup_{\pi,P} S(\pi,P),\|\cdot\|_F,\epsilon\right)
&\le \log \sum_{\pi,P}\mathcal{N}\bigl(S(\pi,P),\|\cdot\|_F,\epsilon\bigr) \notag\\
&\le (k+2\ell)R^2\log(1/\epsilon) + c_1 (k+2\ell)R^2 n + c_2\log(k+2\ell) + (n+2\ell)\log k\notag\\
&\le (k+2\ell)R^2\log(1/\epsilon) + c_1 (k+2\ell)R^2 n + O(n+\ell\log n),
\label{eq:union-F}
\end{align}
where the second inequality is by plugging \eqref{eq:here}, and last inequality follows from $k\leq n$. 

Since $\|A\|_2\le \|A\|_F$, any Frobenius $\epsilon$-cover is a spectral $\epsilon$-cover; thus
\begin{equation}\label{eq:op-le-F}
\log \mathcal{N}\left(\bigcup_{\pi,P} S(\pi,P),\|\cdot\|_2,\epsilon\right)
\le
\log \mathcal{N}\left(\bigcup_{\pi,P} S(\pi,P),\|\cdot\|_F,\epsilon\right).
\end{equation}

On the other hand, by \Cref{lem:suupper} and that $0<2\epsilon\leq 1/10$, we have 
\begin{equation}\label{eq:U-lower}
\log \mathcal{N}\bigl(U(2^n),\|\cdot\|_2,2\epsilon\bigr)
\ge 4^n \log\bigl(3/(8\epsilon)\bigr).
\end{equation}

By the assumption of the lemma, we have that $U(2^n)$ is contained in the
$\epsilon$-neighborhood (in spectral norm) of $\bigcup_{\pi,P} S(\pi,P)$. Then
any $\epsilon$-cover of $\bigcup_{\pi,P} S(\pi,P)$ is a $2\epsilon$-cover
of $U(2^n)$, which implies
\[
\mathcal{N}\bigl(U(2^n),\|\cdot\|_2,2\epsilon\bigr)
\le \mathcal{N}\left(\bigcup_{\pi,P} S(\pi,P),\|\cdot\|_2,\epsilon\right).
\]
Combining this with \eqref{eq:op-le-F} and \eqref{eq:union-F}, and recalling that $R^2=4^{\,n/k+m+O(1)}$, yields
\begin{equation}\label{eq:master}
\log \mathcal{N}\bigl(U(2^n),\|\cdot\|_2,2\epsilon\bigr)
\le (k+2\ell)\,4^{\,n/k+m+O(1)}\bigl(\log(1/\epsilon)+c_1 n\bigr)+O(n+\ell\log n).
\end{equation}
Inserting the lower bound \eqref{eq:U-lower} into \eqref{eq:master} gives
\[
4^n\log\bigl(3/(8\epsilon)\bigr)
\le (k+2\ell)\,4^{\,n/k+m+O(1)}\bigl(\log(1/\epsilon)+c_1 n\bigr)+O(n+\ell\log n).
\]
The left-hand side is $\Theta\left(4^n\log(1/\epsilon)\right)$.
By the assumption that $k\leq n$ and $\epsilon\ge 2^{-n}$, i.e., $\log(1/\epsilon)\leq n$, the right-hand side is
dominated by the term $2\ell\cdot 4^{\,n/k+m+O(1)} n$.
By comparing the two sides, the number $\ell$ of nonlocal gates should satisfy
\[
\ell = \Omega\left(\frac{4^{\,n}\log(1/\epsilon)}{4^{\,n/k+m+O(1)}\,n}\right)
= \Omega\left(\frac{4^{(1-1/k)n-m}\,\log(1/\epsilon)}{n}\right). \qedhere
\]
\end{proof}

\begin{remark}
If we are only concerned with the lower bound
$C_m^{(k)}(U)=\Omega(4^{(1-1/k)n-m})$ on exact communication complexity, a much
simpler proof follows from parameter counting, which is presented in
\Cref{appx:simple}.  
\end{remark}

\subsection{Space-Unbounded Lower Bound}
\label{sec:generallowub}

The last ingredient of \Cref{thm:commk} is a linear lower bound on
communication even when ancillas are unlimited. We actually prove a
stronger result, that there exists a unitary, namely $U_{\rm SIP}$, that
requires $\Omega(n)$ communication to approximate within constant error. 

\begin{definition}[Shifted inner product]
For $i\in[n]$ and $x,y\in\{0,1\}^n$, define
\[
\mathrm{SIP}(i,x,y)=\bigoplus_{j=0}^{n-1} x_j\,y_{(j+i)\bmod n}.
\]
\end{definition}

\begin{lemma}\label{lem:twoomegan}
Define the $(\log n+2n+1)$-qubit unitary that computes $\mathrm{SIP}(i,x,y)$ by
\[
U_{\mathrm{SIP}}\ket{i,x,y,z}:=\ket{i,x,y,\,z\oplus \mathrm{SIP}(i,x,y)}.
\]
Then $\CC_{\infty}\left(U_{\mathrm{SIP}};\epsilon\right)=\Omega(n)$ for any constant $0 \leq \epsilon < 1/\sqrt{2}$.
\end{lemma}

\begin{proof} 

Prove by contradiction. Suppose that $U_{\rm SIP}$ can be
$\epsilon$-approximated (under spectral norm) by a circuit $T$ using $o(n)$
nonlocal two-qubit gates. If we apply $T$ on $\ket{i, x, y}_{Q}\ket{0}_{R}$ and measure register ${R}$, it output $f:={\rm SIP}(i, x, y)$ with error
\[
  \bra{\phi'}\Pi_{\rm err}\ket{\phi'}=
  (\bra{\phi'}-\bra{\phi})\Pi_{\rm err}(\ket{\phi'}-\ket{\phi})\leq
  \|\Pi_{\rm err}\|_2 \|(T-U_{\rm SIP})\ket{i, x,y,0}\|^2 \leq 
  \|T-U_{\rm SIP}\|_2^2 = \epsilon^2.
\]
where $\ket{\phi}=U_{\rm SIP}\ket{i,x,y,0}, \ket{\phi'}=T\ket{i,x,y,0}$, and 
$\Pi_{\rm err}=\mathbbm{I}_Q \otimes \ket{f\oplus 1}\bra{f\oplus 1}$.

Then there is a protocol $P_n$ with the following property: Alice and Bob each
hold $n+\frac{1}{2}\log n+O(1)$ input bits among $(i,x,y)$, and they compute
$\mathrm{SIP}(i,x,y)$ with error $\epsilon^2<1/2$ using $o(n)$ qubits of
communication. Without loss of generality, assume Alice holds at least $n/2$ bits of $x$, 
so Bob holds at least $n/2-\log n+O(1)$ bits of $y$. Define
$A := \{j \mid \text{Alice holds } x_j\}$ and 
$B := \{j \mid \text{Bob holds } y_j\}$, 
so that $|A|,|B|=(\frac{1}{2}-o(1))n$. For any fixed $i\in[n]$, set
$B_i := \{j \mid j+i \in B\}$. 
Observe that
\[
\sum_{i\in[n]} |A \cap B_i|
= \sum_{j\in A} \bigl|\{\,i \mid j\in B_i\,\}\bigr|
= |A|\,|B|
= \frac{n^2}{4}-o(n)>\frac{n^2}{5}.
\]
Hence there exists some $i^\ast$ with $|A\cap B_{i^\ast}|\ge n/5$.

We now build a new protocol $P'$ that computes $\mathrm{IP}(x,y):=\bigoplus_{i=0}^{n-1} x_iy_i$ with $o(n)$
communication. In protocol $P_{5n}$, by
the same averaging, there exists an $i^\ast$ such that $|A\cap B_{i^\ast}|\ge
n$. Given two $n$-bit inputs $x$ and $y$, construct $x',y'\in\{0,1\}^{5n}$ as
follows:
\begin{enumerate}
\item place $x$ and $y$ into the coordinates of $x'$ and $y'$ indexed by $A\cap B_{i^\ast}$, respectively;
\item set all remaining coordinates to $0$.
\end{enumerate}

Run $P_{5n}(i^\ast, x',y')$ and it outputs $\mathrm{SIP}(i^\ast,
x',y')=\mathrm{IP}(x,y)$ with error $<1/2$. Therefore $P'$ computes
$\mathrm{IP}(x,y)$ with $o(n)$ qubits of communication and error $<1/2$,
contradicting \Cref{lem:ip}. We conclude that $C_m(U_{\rm SIP};\epsilon)=\Omega(n)$
regardless of $m$.
\end{proof}

\begin{remark}
The existence of an $n$-qubit unitary $U$ that requires $\Omega(n)$ communication
also follows from
\Cref{lem:cnotlow}, which gives a linear lower bound for CNOT circuits.
However, the lower bound for $U_{\rm SIP}$ is stronger in the sense that it tolerates
error up to $\epsilon<1/\sqrt{2}$, whereas \Cref{lem:cnotlow} only allow error $\epsilon<1/4$.
\end{remark}

\section{Quantum Fourier Transform}\label{sec:QFT}

In this section, we analyze the distributed communication complexity of $n$-qubit QFT, which is defined as
\[
\QFT_n\ket{x_1\cdots x_n}
:=\frac{1}{2^{n/2}}\bigotimes_{j=1}^{n}\left(\ket{0}+e^{2\pi i\,0.x_jx_{j+1}\cdots x_n}\ket{1}\right).
\]

The main result of this section is a tight linear characterization of the communication
complexity of $\QFT_n$, as stated in the following theorem.
\begin{theorem} \label{thm:qft}
    $\CC_m(\QFT_n)=\Theta(n)$ for any integer $m\geq 1$.
\end{theorem}

We also show that \Cref{thm:qft} is not robust against approximation: if we allow
approximately implementating $\QFT_n$ up to inverse polynomial error, then only
logarithmic communication is required, as stated in the following lemma.

\begin{restatable}{lemma}{AQFTLogLemma}\label{lem:aqft-log}
For $\epsilon>0$, we have $\CC_1(\QFT_n;\epsilon)=O(\log(n/\epsilon))$.
\end{restatable}

The constructions in \Cref{thm:qft} and \Cref{lem:aqft-log} naturally extend to the setting with \(k\) processors with an arbitrarily connected topology, achieving \(O(nk)\) and \(O(k\log (n/\epsilon))\) communication, respectively (\Cref{coro:qft,coro:aqft}). 

\begin{remark} \label{rmk:rev}
    For simplicity of presentation,
    $\QFT_n$ is defined without the bit reversal step at the final stage.
    All the bounds in this section still hold for the
    standard QFT with bit reversal; see \Cref{sec:reversal} for details. Thus,
    our model definition is robust to this variation.
\end{remark}

\subsection{Upper Bound for Exact QFT}

We first prove the upper bound of \Cref{thm:qft}, which states that $O(n)$
communication is sufficient to exactly implement $\QFT_n$ across two processors. Formally,

\begin{lemma}[Upper bound] \label{lem:qftup}
    $\CC_1(\QFT_n)\leq n$.
\end{lemma}

The standard circuit for $\QFT_n$ uses $O(n^2)$ controlled-rotation gates. Under
any balanced bipartition of the qubits into two processors, the circuit yields
$O(n^2)$ nonlocal controlled rotation gates. However, by batching all controlled
rotations that share the same target qubit and performing them after
communicating that target once across the partition, the total communication cost can be reduced to
$O(n)$. Similar strategies have appeared in the literature
\cite{neumann2020imperfect}; we include the construction here for completeness.

\begin{proof}[Proof of \Cref{lem:qftup}]
First assign the first $n/2$ qubits and the last $n/2$ qubits to processors $A$ and $B$, respectively. Observe that the 2-qubit controlled-rotation gates in ${\rm QFT}_n$ can be grouped by their \emph{target} qubit into $n-1$ sets $S_1,\ldots,S_{n-1}$, as illustrated below.
\begin{figure}[htbp]
    \centering
    \begin{tikzpicture}
    \yquantdefinebox{dots}[inner sep=0pt]{$\dots$}
    \yquantdefinebox{vdotsqft}[inner sep=0pt]{$\vdots$}    
    \yquantdefinebox{ghostqft}[inner sep=0pt]{$\quad\quad$}
    \begin{yquant}
        qubit {$\ket{j_1}$} j1;
        qubit {$\ket{j_2}$} j2;
        nobit ellipsis;
        vdotsqft ellipsis;
        qubit {$\ket{j_{n-1}}$} jn1;
        qubit {$\ket{j_{n}}$} jn;
        h j1;
        [this subcircuit box style={dashed, "$S_1$"}]
        subcircuit {
            qubit {} j1;
            qubit {} j2;
            nobit ellipsis;
            qubit {} jn1;
            qubit {} jn;
            
            box {$R_2$} j1 | j2;
            ghostqft ellipsis;
            dots ellipsis;
            box {$R_{n-1}$} j1 | jn1;
            box {$R_{n}$} j1 | jn;
        } (j1, j2, ellipsis, jn1, jn);
        h j2;
        [this subcircuit box style={dashed, "$S_2$"}]
        subcircuit {
            qubit {} j1;
            qubit {} j2;
            nobit ellipsis;
            qubit {} jn1;
            qubit {} jn;
            dots ellipsis;
            box {$R_{n-2}$} j2 | jn1;
            box {$R_{n-1}$} j2 | jn;
        } (j1, j2, ellipsis, jn1, jn);
        dots ellipsis;
        h jn1;
        [this subcircuit box style={dashed, "$S_{n-1}$"}]
        subcircuit {
            qubit {} j1;
            qubit {} j2;
            nobit ellipsis;
            qubit {} jn1;
            qubit {} jn;
            box {$R_{2}$} jn1 | jn;
        }  (j1, j2, ellipsis, jn1, jn);
        h jn;
    \end{yquant}
\end{tikzpicture}
    \caption{The QFT circuit and gate grouping strategy.}
    \label{fig:qft}
\end{figure}

Among them, the groups $S_{n/2+1},\ldots,S_{n-1}$ require no nonlocal two-qubit gates. For each $S_i$ with $1\le i\le n/2$, all gates share the $i$-th qubit as the target. 

As illustrated in \Cref{fig:si}, the $i$-th group can therefore be implemented using only two nonlocal SWAPs: first apply all two-qubit gates in $S_i$ that are local; then SWAP the $i$-th qubit to an ancilla on processor $B$; next implement the remaining two-qubit gates of $S_i$ using the ancilla as the target qubit; finally SWAP the $i$-th qubit back to its original location. In total, this uses at most $n$ nonlocal two-qubit gates.
\end{proof}

\begin{figure}
    \centering
    \begin{tikzpicture}
    \yquantdefinebox{dotss}[inner sep=0pt]{$\dots$}
    \yquantdefinebox{dotsv}[inner sep=0pt]{$\vdots$}
    \yquantdefinebox{ghost}[inner sep=0pt]{$\quad\quad\quad\quad$}
    \begin{yquant}
        qubit {$\ket{j_1}$} j1;
        nobit e0;
        
        qubit {$\ket{j_{i}}$} ji;
        qubit {$\ket{j_{i+1}}$} ji1;
        qubit {$\ket{j_{i+2}}$} ji2;
        nobit e1;
        qubit {$\ket{j_{\frac{n}{2}}}$} jn2;        

        qubit {$\ket{0}$} aux1;
        qubit {$\ket{j_{\frac{n}{2}+1}}$} jn2p1;
        nobit e2;
        dotsv e2;
        qubit {$\ket{j_{n-1}}$} jn1;
        qubit {$\ket{j_{n}}$} jn;

        [this subcircuit box style={dashed, "Local computation"}]
        subcircuit {
            qubit {} j1;
            nobit e0;
            
            qubit {} ji;
            qubit {} ji1;
            qubit {} ji2;
            nobit e1;
            qubit {} jn2;

            dotsv e0;
            dotsv e1;
            box {$R_2$} ji | ji1;
            box {$R_3$} ji | ji2;
            ghost e1;    
            dotss e1;
            box {$R_{n/2-i+1}$} ji | jn2;    
        } (j1, e0, ji, ji1, ji2, e1, jn2);

        swap (ji, aux1);

        [this subcircuit box style={dashed, "Local computation"}]
        subcircuit {
            qubit {} aux1;
            qubit {} jn2p1;
            nobit e2;
            qubit {} jn1;
            qubit {} jn;

            box {$R_{n/2-i+2}$} aux1 | jn2p1;
            ghost e2;
            dotss e2;
            box {$R_{n-i}$} aux1 | jn1;
            box {$R_{n-i+1}$} aux1 | jn;
        } (aux1, jn2p1, e2, jn1, jn);

        swap (ji, aux1);

        output {$A$} (j1, e0, ji, ji1, ji2, e1, jn2);
        output {$B$} (aux1, jn2p1, e2, jn1, jn);
    \end{yquant}
\end{tikzpicture}
    \caption{Implementation of gate group $S_i$ across two processors.}
    \label{fig:si}
\end{figure}

The above construction can be easily generalized to $k$  processors:
\begin{corollary} \label{coro:qft}
Given any connected undirected graph $G=([k],E)$, we have $\CC_{1}^{G}(\QFT_n)=O(kn)$.

\end{corollary}

\begin{proof}
Let $H$ be a spanning tree of $G$. Relabel the processors so that $1,2,\ldots,k$ is a depth-first-search (DFS) preorder traversal of $H$. 
For $i\in[k-1]$, let $d_i$ be the graph distance in $H$ between processors $i$ and $i{+}1$.  
A standard property of DFS on trees gives
\(
\sum_{i=1}^{k-1} d_i \le 2k-2 = O(k).
\)
Assign data qubits $q_{(\ell-1)(n/k)+1},\ldots,q_{\ell(n/k)}$ to processor $P_\ell$ for $\ell\in[k]$.

It suffices to show that each gate group $S_i$ in \Cref{fig:qft} can be implemented over $H$ using $O(k)$ nonlocal two-qubit gates, assuming one ancilla qubit per processor.  
Let the target of $S_i$ reside on $P_\ell$, where $\ell=\lceil i/(n/k)\rceil$.  
Starting at $P_\ell$, move (via SWAP) the target state along the DFS order
$P_\ell \to P_{\ell+1} \to \cdots \to P_k$, keeping it on the local ancilla at each visited processor and, at each step, executing all controlled rotations in $S_i$ whose controls are local.  
Then route the target back to its original location at $P_\ell$.

Between consecutive processors $P_j$ and $P_{j+1}$ in the DFS order, routing the target requires $d_j$ nonlocal SWAPs (one per edge on the unique path in $H$).  
The forward-and-back traversal therefore uses 
\(
2\sum_{j=1}^{k-1} d_j = O(k)
\)
nonlocal SWAPs per group $S_i$.  
Since there are $O(n)$ groups, the total communication is $O(nk)$.

\end{proof}

\subsection{Lower Bound for Exact QFT}

We prove that linear communication is necessary to exactly implement $\QFT_n$
over two processors even with unlimited ancillas, so the construction in
\Cref{lem:qftup} is optimal up to a constant factor.

\begin{lemma}[Lower bound] \label{lem:qftlow}
    $\CC_\infty(\QFT_{2n})=\Omega(n)$.
\end{lemma}

To prove the lower bound, we use the following lemma, which links the number of
communication qubits to the rank of the joint output probability matrix. This
rank method goes back to Theorem 4.2 of \cite{zhang2012quantum}, which was
proved with constant factor $1/4$. Here, we provide a tightened version of
constant $1/2$, of which a self-contained proof is given in \Cref{appx:rank}.

\begin{lemma} \label{lem:rank}
    Suppose two parties $A$ and $B$ communicate $c$ qubits. Let $x$ and $y$ denote Alice and Bob's outputs respectively. Define matrix \[M_{x,y}:=\Pr(A\text{ outputs }x, B\text{ outputs }y).\] 
    Then $c\geq \frac{1}{2}\log {\rm rank}(M).$
\end{lemma}

\paragraph{Special case: fixed partition.}
To illustrate the idea behind the proof of \Cref{lem:qftlow}, we first consider the special case where the qubit partition is fixed as 
\[
\pi^*(i):=\begin{cases}
    1 &\text{if } 1\leq i\leq n \\
    2 &\text{if } n < i \leq 2n
\end{cases}.
\]
Formally, we claim that
\begin{lemma} \label{lem:pistar}    
    $\CC_\infty(\QFT_{2n}\mid\pi^*)\geq \frac{n-1}{4}$.
\end{lemma}

Given a implementation of $\QFT_{2n}$ under $\pi^*$, let two parties Alice and Bob
run the following experiment, and output $n$-bit strings $x$ and $y$ respectively.
\begin{enumerate}
    \item Alice samples a uniform $x\in\{0,1\}^n$. 
    \item Bob holds $\ket{0^n}$, Alice holds $\ket{x}$. Apply ${\rm QFT}_{2n}$ and Bob gets
   $$
   \bigotimes_{i=1}^n \frac{\ket{0}+e^{2\pi i\frac{x}{2^{2n-i}}}\ket{1}}{\sqrt{2}}.
   $$
   \item Bob applies transversal $H$ on its register and gets
   $$
   \bigotimes_{i=1}^n \left(\frac{1+e^{2\pi i\frac{x}{2^{2n-i}}}}{2}\ket{0}+\frac{1-e^{2\pi i\frac{x}{2^{2n-i}}}}{2}\ket{1}\right).
   $$
   \item Bob measures all qubits and obtains $y\in\{0,1\}^n$:
   $$
   \Pr(y_i=0)=\left|\frac{1+e^{2\pi i\frac{x}{2^{2n-i}}}}{2}\right|^2=\cos^2 \frac{\pi x}{2^{2n-i}}.
   $$
\end{enumerate}
Define matrix $M_{x,y}:=\Pr(A\text{ outputs }x,\, B\text{ outputs }y)$.
By the above protocol, the $x$-th row of $M_{x,y}$ is
\[
M_{x,\star}=\frac{1}{2^n}\bigotimes_{i=1}^n \left[\cos^2 \frac{\pi x}{2^{2n-i}}, \sin^2 \frac{\pi x}{2^{2n-i}}\right].
\]
Then we claim that 
\begin{lemma} \label{lem:m}
    ${\rm rank}(M)\geq 2^{n-1}$.
\end{lemma}

\begin{proof}
By the fact that
$[\cos^2 \theta,\sin^2\theta]=[1, \cos2\theta]\begin{bmatrix}
1/2&1/2\\1/2&-1/2
\end{bmatrix},
$
we have
$$
M_{x,\star}=\left(\bigotimes_{i=1}^n \left[1, \cos \frac{2\pi x}{2^{2n-i}}\right]\right)\cdot \frac{1}{2^n}\begin{bmatrix}
1/2&1/2\\1/2&-1/2
\end{bmatrix}^{\otimes n}.
$$
Since the latter is invertible, ${\rm rank}(M)$ equals the rank of the matrix
$$
M'_{x,\star}:=\bigotimes_{i=1}^n \left[1, \cos \frac{2\pi x}{2^{2n-i}}\right].
$$
Let $a_x=2\pi x/2^{2n}$, and $T_k$ be the Chebyshev polynomial of degree $K$ \cite{mason2002chebyshev}. By the fact that $\cos(k\alpha)=T_k(\cos \alpha)$, we have
$$
M'_{x,y}=\prod_{i=1}^n\cos^{y_i}\left(\frac{2\pi x}{2^{2n-i}}\right)=\prod_{i=1}^n\cos^{y_i}\left(2^ia_x\right)=\prod_{i=1}^n T_{2^i}(\cos a_x)^{y_i}=f_y(\cos a_x)
$$
where $f_y$ is some degree-$y$ polynomial. 

Thus the $y$-th column $M'_{\star,y}$ is the evaluation of $f_y$ on $2^n$ points
$$
\{\cos a_x\}=\left\{\cos(0), \cos\left(2\pi\frac{1}{2^n}\right), \ldots, \cos\left(2\pi\frac{2^n-1}{2^n}\right)\right\}
$$
which contains $\geq 2^n/2$ distinct points (since $\cos \theta$ is monotone for $\theta\in[0,\pi]$).

Finally we show that the first $2^{n-1}$ columns $M'_{\star,0}, \ldots, M'_{\star,2^{n-1}-1}$ are linearly independent by contradiction. Assume there exists $\gamma_1,\ldots,\gamma_k\not=0$ and $0\leq y^{(1)}<y^{(2)}<\cdots<y^{(k)}<2^{n-1}$ such that 
$$
\sum_{i=1}^k \gamma_i M'_{\star,y^{(i)}}=\mathbf{0}.
$$
Then the degree-$y^{(k)}$ polynomial
$$
F:=\sum_{i=1}^k \gamma_i f_{y^{(i)}}
$$
evaluates $0$ on $2^n/2$ distinct points. Thus $F\equiv 0$ which implies $\gamma_k=0$, contradicting with the fact that $\gamma_k\not=0$. Thus ${\rm rank}(M)={\rm rank}(M')\geq 2^{n-1}$.
\end{proof}

Now we conclude lower bound for the special case.
\begin{proof}[Proof of \Cref{lem:pistar}]
    By \Cref{lem:rank} and \Cref{lem:m}, the number of communication qubits 
    $$
    c\geq \frac{1}{2}\log{\rm rank}(M)\geq \frac{n-1}{2}.
    $$
    By the fact that an nonlocal qubit gate can be implemented by communication two qubits, we have that $\CC_m(\QFT_{2n}\mid\pi^*)\geq \frac{n-1}{4}$ regardless of $m$.
\end{proof}

\paragraph{General case.} The above proof can be easily generalized to arbitrary balanced partition.

\begin{proof}[Proof of \Cref{lem:qftlow}]
    Given any balanced paration $\pi$ of $2n$ input qubits, by an averaging
    argument, we can assume Alice holds a set $X$ of $n/2$ qubits among
    $\{n+1,\ldots,2n\}$, and Bob holds a set $Y$ of $n/2-O(1)$ qubits among
    $\{1, \ldots, n\}$. The experiment protocol is almost the same:
    \begin{enumerate}
        \item Alice samples $x\in\{0,1\}^{n/2}$.
        \item Alice initializes $X$ as $\ket{x}$, and all other qubits are set to $\ket{0}$. Then they apply $\QFT_{2n}$.
        \item Bob applies transversal $H$ and measurement on $Y$,  obtaining $y\in\{0,1\}^{n/2-O(1)}$.
    \end{enumerate}
    Define matrix $M_{x,y}:=\Pr(A\text{ outputs }x,\, B\text{ outputs }y)$.
    By the same argument, we can show ${\rm rank}(M)=\Omega(2^n)$, which implies $\CC_m(\QFT_{2n})=\Omega(n)$ for any $m$.
\end{proof}

Finally, we conclude \Cref{thm:qft}.
\begin{proof}[Proof of \Cref{thm:qft}]
    Combine \Cref{lem:qftup} and \Cref{lem:qftlow}.
\end{proof}

\subsection{Approximate Communication Complexity of QFT}

We have proved that \emph{exactly} implementing $\QFT_n$ requires strictly
linear communications. However, if we allow \emph{approximate} implementation,
the communication cost can be exponentially reduced, as characterized by \Cref{lem:aqft-log}, which we restate below.

\AQFTLogLemma*

\begin{proof}
By using approximate QFT, we explicitly construct a two-processor implementation using one ancilla qubit per side that
(i) uses $O(\log (n/\epsilon))$ nonlocal two-qubit gates and (ii) approximates $\QFT_n$ with error $\epsilon$.

Consider the standard implementation of $\QFT_n$ as a sequence of $n$ single-qubit Hadamard gates
and controlled rotation gate ${\rm CR}_{2\pi/2^{d}}$ between qubits with $d\in\{1,\dots,n-1\}$.
For a parameter $b\ge 1$, define $U_b$ by keeping all Hadamards and those ${\rm CR}$ gates with $d\le b$, and dropping all ${\rm CR}$ with $d>b$. Let $A$ and $B$ hold the first and second half of the input qubits respectively. 
Similar to \Cref{fig:qft}, we group controlled-rotations by their target qubit into $n-1$ sets $S_1,\ldots,S_{n-1}$, where $S_i$ contains all the controlled rotation gates with target qubit $i$. 
Observe that each group $S_i$ contains only $b$ CR gates, whose control qubits range consecutively from $i+1$ to $i+b$. Thus at most $b$ group $S_i$'s are nonlocal (i.e., contain nonlocal gates). Each such nonlocal $S_i$ can be implemented using two nonlocal SWAP operations and one ancilla qubit, following the strategy illustrated in \Cref{fig:si}. Thus at most $2b$ 
nonlocal SWAP gates are required to implement $U_b$ across $A$ and $B$.

Then we analyze the approximation error. Let the exact circuit be $\prod_t G_t$ and the truncated one $\prod_t \widetilde G_t$, where $\widetilde G_t\in\{G_t,\mathbbm{I}\}$. Write the gates in the same order as $G_L\cdots G_1$ and $\widetilde G_L\cdots \widetilde G_1$. By the telescoping identity,
\[
G_L\cdots G_1-\widetilde G_L\cdots \widetilde G_1
=\sum_{s=1}^L G_L\cdots G_{s+1}(G_s-\widetilde G_s)\widetilde G_{s-1}\cdots \widetilde G_1.
\]
Taking spectral norm and using unitary invariance ($\|UXV\|_2=\|X\|_2$ for unitaries $U,V$), we obtain
\[
\|\QFT_n-U_b\|_2
\le \sum_t \|G_t-\widetilde G_t\|_2
= \sum_{\text{dropped }{\rm CR}_\theta}\|{\rm CR}_\theta-\mathbbm{I}\|_2.
\]
For a dropped controlled rotation ${\rm CR}_{\theta}=\operatorname{diag}(1,1,1,e^{i\theta})$ with $\theta=2\pi/2^d$, we have
\[
\|{\rm CR}_{\theta}-\mathbbm{I}\|_2=|e^{i\theta}-1|\le \theta = \frac{2\pi}{2^d}.
\]
There are exactly $n-d$ controlled rotations with angle $2\pi/2^d$, so
\[
\|\QFT_n-U_b\|_2
\le\ \sum_{d>b}(n-d)\,\frac{2\pi}{2^{d}}
\le\ n\sum_{d>b}\frac{2\pi}{2^d}
\le 2\pi n\,2^{-b}.
\]
To achieve approximation error $\epsilon$, we need $2\pi n 2^{-b} \leq
\epsilon$, which implies $b=O(\log (n/\epsilon))$. Thus $\CC_1(\QFT_n;\epsilon)=O(b)=O(\log (n/\epsilon))$. \qedhere

\end{proof}

The above construction can be easily generalized to $k$ processors:
\begin{corollary} \label{coro:aqft}
Given any connected undirected graph $G=([k],E)$, we have $\CC_{1}^{G}(\QFT_n;\epsilon)=O(k\log(n/\epsilon))$.

\end{corollary}
\begin{proof}
Let $H$ be a spanning tree of $G$. Relabel the processors so that $1,2,\ldots,k$ is a depth–first–search (DFS) preorder of $H$.  
For $j\in[k-1]$, let $d_j$ be the graph distance in $H$ between processors $j$ and $j{+}1$.  
A standard property of DFS on trees gives
\(
\sum_{j=1}^{k-1} d_j \le 2k-2 = O(k).
\)

We show that the approximate QFT with parameter $b$ (as defined in the proof of \Cref{lem:aqft-log}) can be implemented using $O(bk)$ nonlocal two-qubit gates over $H$.  
Assign qubits $q_{(\ell-1)(n/k)+1},\ldots,q_{\ell(n/k)}$ to processor $P_\ell$ for $\ell\in[k]$. Note that the gate group $S_i$ contains $b$ CR gates, whose target qubit is $q_i$ and control qubit ranges from $q_{i+1}$ to $q_{i+b}$. Thus each $S_i$ spans at most $r:=\lceil bk/n\rceil+1$ consecutive processors in the DFS order. Then use the construction in \Cref{coro:qft} to implement each $S_i$ across the corresponding $r$ consecutive processors.

We analyze the communication cost of our construction: For each adjacent pair $(P_\ell, P_{\ell+1})$, at most $b$ group $S_i$'s cross the pair, and each such $S_i$ contribute $2d_k$ nonlocal SWAPs. Thus the total number of nonlocal SWAPs are $2b\sum_{\ell= 1}^{k-1}d_k=O(bk)$.


Finally, setting $b = \log(n/\epsilon)$ yields the desired $O(k\log(n/\epsilon))$ bound.
\end{proof}

\subsection{Extending to QFT with Bit Reversal}
\label{sec:reversal}

Throughout this section, we have focused on the QFT unitary without the final
bit-reversal step. However, all upper and lower bounds also extend to the
conventional QFT unitary
\[
\QFT_n^{\uparrow}\ket{x}:=
\frac{1}{2^{n/2}}\sum_{y\in\{0,1\}^n} e^{ \frac{2\pi i x y}{2^n}}\ket{y}
=\frac{1}{2^{n/2}}\bigotimes_{j=1}^{n}\left(\ket{0}+e^{2\pi i\,0.x_{n-j+1}x_{n-j+2}\cdots x_n}\ket{1}\right),
\]
whose standard implementation is illustrated in \Cref{fig:qftrev}.

\begin{figure}[htbp]
    \centering
    \begin{tikzpicture}
    \yquantdefinebox{dots3}[inner sep=0pt]{$\dots$}
    \yquantdefinebox{vdotsqft3}[inner sep=0pt]{$\vdots$}    
    \yquantdefinebox{ghostqft3}[inner sep=0pt]{$\quad\quad$}
    \begin{yquant}
        qubit {$q_1$} x1;
        qubit {$q_2$} x2;
        nobit ellipsis;
        vdotsqft3 ellipsis;
        qubit {$q_{n-1}$} xn1;
        qubit {$q_{n}$} xn;

        [this subcircuit box style={dashed, "$\QFT_n$"}]
        subcircuit {
            qubit {} x1;
            qubit {} x2;
            nobit ellipsis;
            qubit {} xn1;
            qubit {} xn;

            h x1;
            
            box {$R_2$} x1 | x2;
            ghostqft3 ellipsis;
            dots3 ellipsis;
            box {$R_{n-1}$} x1 | xn1;
            box {$R_{n}$} x1 | xn;

        h x2;
        [this subcircuit box style={dashed, "$S_2$"}]

        box {$R_{n-2}$} x2 | xn1;
        box {$R_{n-1}$} x2 | xn;

        dots3 ellipsis;
        h xn1;

        box {$R_{2}$} xn1 | xn;

        h xn;

        } (x1, x2, ellipsis, xn1, xn);

        [this subcircuit box style={dashed, "Bit reversal"}]
        subcircuit {
            qubit {} x1;
            qubit {} x2;
            nobit ellipsis;
            qubit {} xn1;
            qubit {} xn;

            swap (x1, xn);
            swap (x2, xn1);
            dots3 ellipsis;

        } (x1, x2, ellipsis, xn1, xn);
    \end{yquant}
\end{tikzpicture}
    \caption{Circuit illustration of $\QFT_n^{\uparrow}$,
    obtained by appending a bit-reversal stage to $\QFT_n$.
    The bit reversal consists of $\lfloor n/2\rfloor$ SWAP gates,
    applied between $(q_1, q_n), (q_2, q_{n-1}), \ldots, (q_{\lfloor n/2\rfloor}, q_{\lceil n/2\rceil})$.
    } \label{fig:qftrev}
\end{figure}

All upper and lower bounds in this section still hold for
$\QFT_n^{\uparrow}$, after minor adjustments to the input partition.
We briefly explain why.
\begin{itemize}
    \item Upper bounds.
    \begin{enumerate}
        \item $\CC_1(\QFT_n^{\uparrow})=O(n)$ and $\CC_1^G(\QFT_n^{\uparrow})=O(kn)$: The bit reversal consists of $O(n)$ SWAP gates, so it contributes at most $O(n)$ communication. Hence the bounds remain unchanged.
        \item $\CC_1(\QFT_n^{\uparrow};\epsilon)=O(\log(n/\epsilon))$: 
            Choose the input partition as 
            \[
            \pi(i):=\begin{cases}
                1 &\text{if } i< n/4 \text{ or } i>3n/4 \\
                2 &\text{if } n/4 \leq i \leq 3n/4
            \end{cases}.
            \]
            Then all SWAPs in the bit-reversal stage are local and require no
            communication. The partition $\pi$ creates two cuts in
            $\{1,\ldots,n\}$, at $n/4$ and $3n/4$. Each cut contributes
            $O(\log(n/\epsilon))$ communication by the same argument as in
            \Cref{lem:aqft-log}, so the total cost is $O(\log(n/\epsilon))$.
        \item $\CC_1^G(\QFT_n^{\uparrow};\epsilon)=O(k\log(n/\epsilon))$: 
        Rearrange the input qubits as follows (which has no cost since we can choose input partition)
        \[
        q_1,\ldots,q_n \gets
        q_1,\ldots,q_{n/4}, q_{3n/4+1},\ldots,q_n, q_{n/4+1},\ldots,q_{3n/4}.
        \]
        Then apply the same argument as in \Cref{coro:aqft}.
    \end{enumerate}
    \item Lower bound $\CC_\infty(\QFT_{2n}^{\uparrow})=\Omega(n)$: Observe that
    the bit-reversal step swaps Alice's and Bob's output registers,
    which merely transposes the output probability matrix $M$ and therefore preserves its rank.
\end{itemize}

\section{Clifford Circuits}\label{sec:Clifford}

In this section, we give a tight characterization for the distributed communication
complexity of Clifford circuits, which are circuits generated only by
$\{H,S,{\rm CNOT}\}$. It is well-known that implementing $n$-qubit Clifford
circuits requires $\tilde{\Theta}(n^2)$ gates \cite{jiang2020optimal}, so
directly partition the circuit into two processors yields $\tilde{O}(n^2)$
nonlocal two-qubit gates. However, we show that $O(n)$ nonlocal two-qubit gates
are sufficient to implement any Clifford circuits, and the bound is tight up to
a constant factor, as stated in the following theorem.

\begin{theorem} \label{thm:cnot}
We have that
    \begin{enumerate} 
        \item[(i)] for any $n$-qubit Clifford circuit $T$, $\CC_1(T)=O(n)$;
        \item[(ii)] there exists $n$-qubit Clifford circuit $T$ such that $\CC_{\infty}(T)=\Omega(n)$.
    \end{enumerate}
\end{theorem}

The \(O(n)\) construction extends naturally to \(k\) processors, achieving 
\(O(nkD)\) communication with one ancilla qubit per processor, 
where \(D\) denotes the diameter of the topology graph
(\Cref{coro:cnot}). 
As an immediate corollary, Clifford circuits can be implemented with 
\(O(nk)\) communication when the \(k\) processors are pairwise connected. 
In contrast to the QFT, the \(\Omega(n)\) lower bound for Clifford circuits 
is robust under approximation, as shown in \Cref{lem:cnotlow}.

\subsection{Upper Bound for Clifford Circuits}

In this subsection, we design a distributed implementation of any $n$-qubit
Clifford circuit using only $O(n)$ nonlocal two-qubit gates. We first consider
implementing a special kind of Clifford circuits, called \emph{DAG CNOT circuits}.

\begin{definition}[DAG CNOT circuit]
    Given a directed acyclic graph (DAG) $\vec{G}=([n],E)$, it specifies a CNOT
    circuit $T_{\vec{G}}$, where the $n$ vertices correspond to $n$ bits. Each directed
    edge $(i,j)\in E$ represents a CNOT with control $i$ and target $j$, and all
    edges are executed according to a topological ordering of $\vec{G}$.

Although $\vec{G}$ may admit multiple topological orders, CNOT gates that
act on the same control or on the same target commute, so the action of $T_{\vec{G}}$ is
well-defined. 
\end{definition}

Recall that any $n$-qubit CNOT circuit $T$ has a matrix representation $M\in
\mathbb{F}_2^{n\times n}$ such that $T\ket{x}=\ket{Mx}$ for any $x\in\{0,1\}^n$
(see \Cref{sec:circuit}). Given a DAG $\vec{G}=([n],E)$ with topological order
$1,\ldots,n$, let $T_{\vec{G}}$ be its associated DAG CNOT circuit. It is easy to see
that the matrix representation of $T_{\vec{G}}$ is lower triangular. Conversely, given
any lower triangular matrix $M\in\mathbb{F}_2^{n\times n}$, one can construct a
DAG CNOT circuit whose matrix representation is $M$. 

The following lemma constructs any DAG CNOT circuit with linear communication.
\begin{lemma}\label{lem:dagcomm}
Given a DAG $\vec{G}=([n],E)$, let $T_{\vec{G}}$ be its associated DAG CNOT circuit. Then $\CC_{1}(T_{\vec{G}})\le n/2$.
\end{lemma}

\begin{proof}
Let the topological order of $G$ be $q_1,\ldots,q_n$, and let $M$ denote the matrix representation of $T_{\vec{G}}$ under this ordering. Assign $q_1,\ldots,q_{n/2}$ to processor $A$ and $q_{n/2+1},\ldots,q_n$ to processor $B$. Since $M$ is invertible lower triangular, it can be block-decomposed as
\[
M = \begin{pmatrix} M_{AA} & 0\\ M_{BA} & M_{BB}\end{pmatrix}=
\begin{pmatrix} M_{AA} & 0\\ 0 & \mathbbm{I}\end{pmatrix}
\begin{pmatrix} \mathbbm{I} & 0\\ M_{BA} & \mathbbm{I}\end{pmatrix}
\begin{pmatrix} \mathbbm{I} & 0\\ 0 & M_{BB}\end{pmatrix},
\]
where $\mathbbm{I}$ denotes identity matrix of size $\frac{n}{2}\times\frac{n}{2}$.

Note that $\begin{pmatrix} M_{AA} & 0\\ 0 & \mathbbm{I}\end{pmatrix}$ and $\begin{pmatrix} \mathbbm{I} & 0\\ 0 & M_{BB}\end{pmatrix}$ correspond to a local CNOT circuit acting on $A$ and $B$ respectively. Thus it suffices to implement $\begin{pmatrix} \mathbbm{I} & 0\\ M_{BA} & \mathbbm{I}\end{pmatrix}$ with $n/2$ nonlocal CNOT gates: For every $i\in[n/2]$, observe that the net action of $q_1,\ldots,q_{n/2}$ on $q_{i+n/2}$ is
\[
\ket{q_{i+n/2}}\ \longmapsto\ \left|q_{i+n/2}\oplus M_{BA}(i,1)q_1\oplus\cdots\oplus M_{BA}(i,n/2)q_{n/2} \right\rangle,
\]
where $M_{BA}(x,y)$ denote the $(x,y)$-th entry of $M_{BA}$.
It can be implemented as follows:
\begin{enumerate}
    \item[(i)] compute $r=M_{BA}(i,1)q_1\oplus\cdots\oplus M_{BA}(i,n/2)q_{n/2}$ on an ancilla $a$ of $A$ by a local CNOT circuit.
    \item[(ii)] apply a CNOT between $a$ and $q_{i+n/2}$ to add $r$ into $q_{i+n/2}$;
    \item[(iii)] run the circuit in step (i) again to uncompute $a$ back to its initial state. 
\end{enumerate}
Each $i$ requires $1$ nonlocal CNOT gate, hence $n/2$ nonlocal two-qubit gates are used in total.
\end{proof}

Next, we generalize the construction to any CNOT circuit.

\begin{lemma}\label{lem:cnotcomm}
For any $n$-qubit CNOT circuit $T$, one has $\CC_{1}(T)\le 2n$.
\end{lemma}

\begin{proof}
Let $M\in\mathbb{F}_2^{n\times n}$ be the matrix representation of $T$, and take a PLU decomposition $M=PLU$, where $P$ is a permutation matrix, $L$ is lower triangular, and $U$ is upper triangular. Assign qubits $q_1,\ldots,q_{n/2}$ to $A$ and $q_{n/2+1},\ldots,q_n$ to $B$. The permutation $P$ can be implemented with $n$ SWAPs, thus contributing $\leq n$ nonlocal SWAPs. The matrix $L$ corresponds to a DAG CNOT circuit with topological order $1,\ldots,n$, which by \Cref{lem:dagcomm} uses at most $n/2$ nonlocal two-qubit gates; similarly, $U$ corresponds to a DAG CNOT circuit with topological order $n,\ldots,1$, also requiring at most $n/2$ nonlocal two-qubit gates. Hence $T$ can be implemented using at most $2n$ nonlocal two-qubit gates.
\end{proof}

A direct corollary is that any Clifford circuit can be implemented with $O(n)$ communication.
\begin{corollary} \label{coro:cliffordup}
For any $n$-qubit Clifford circuit $T$, one has $\CC_{1}(T)\le 10n$.
\end{corollary}

\begin{proof}
    By \Cref{lem:clifford}, any Clifford circuit $T$ can be implemented by an 11-layer sequence
    H-C-S-C-S-C-H-S-C-S-C, where $\mathrm{H}$ denotes a layer of $H$ gates,
    $\mathrm{S}$ denotes a layer of $S$ gates, and $\mathrm{C}$ denotes a CNOT circuit. Since $H$ and $S$ layers require no nonlocal
    gates, each CNOT circuit requires $2n$ nonlocal gates by \Cref{lem:cnotcomm}, and they share the same partition. Thus $10n$ nonlocal two-qubit gates are required in total.
\end{proof}

We further generalize the construction to $k$ processors with arbitrary topology.
\begin{lemma} \label{coro:cnot}
Let $H$ be an undirected graph specifying the topology of $k$ processors, and $D$ denote the diameter of $H$.  
For any $n$-qubit Clifford circuit $T$, 
we have $\CC_{1}^{H}(T)=O(nkD)$.
\end{lemma}

\begin{proof}
    By the above analysis, $T$ can be decomposed as constant number of (i) single-qubit gate layers, (ii) permutations, and (iii) DAG CNOT circuits. 
    Item (i) does not need communication. For (ii), each permutation can be realized with at most $n$ SWAPs, and implementing each SWAP needs at most $D$ hops in $H$, contributing $O(nD)$ communication in total. 
    Thus it remains to show that any DAG CNOT circuit can be implemented with $O(nkD)$ nonlocal two-qubit gates.
    
Given any DAG $\vec{G}=([n],E)$, let $T_{\vec{G}}$ be the associated DAG CNOT circuit.
Let the topological order be $q_1,\ldots,q_n$, and let $M$ denote the matrix representation of $T_{\vec{G}}$ in this basis.
Partition the qubits contiguously across $k$ processors: assign $q_{(\ell-1)n/k+1},\ldots,q_{\ell n/k}$ to the $\ell$-th processor $P_\ell$ for $\ell\in[k]$.
With this block layout, write $M_{P_iP_j}$ for the $(i,j)$ block (rows of $P_i$, columns of $P_j$).
Since $M$ is invertible and lower triangular in this ordering, it admits the block factorization
\[
M
=
\left(
\begin{array}{cccc}
M_{P_1P_1} & 0 & \cdots & 0 \\
M_{P_2P_1} & M_{P_2P_2} & \ddots & \vdots \\
\vdots & \ddots & \ddots & 0 \\
M_{P_kP_1} & M_{P_kP_2} & \cdots & M_{P_kP_k}
\end{array}
\right)
=
\left(
    D_1
    \prod_{i=2}^{k} L_{i1}
\right)
\left(
    D_2
    \prod_{i=3}^{k} L_{i2}
\right)
\cdots
\left(
    D_{k-1} L_{k,k-1}
\right)
D_k,
\]
where, for each $t\in[k]$ and $1\le j<i\le k$,

 \[
 D_t=\begin{blockarray}{ccccccc}
P_1 & \cdots & P_{t-1} & P_t & P_{t+1} & \cdots & P_k \\
\begin{block}{(ccccccc)}
\mathbbm{I} &        &        &        &        &        &        \\
             & \ddots &        &        &        &        &        \\
             &        & \mathbbm{I} &        &        &        &        \\
             &        &        & M_{P_tP_t} &        &        &        \\
             &        &        &        & \mathbbm{I} &        &        \\
             &        &        &        &        & \ddots &        \\
             &        &        &        &        &        & \mathbbm{I} \\
\end{block}
\end{blockarray}
,\ L_{i,j}=
\begin{blockarray}{ccccccc}
P_1 & \cdots & P_j & \cdots & P_i & \cdots & P_k \\
\begin{block}{(ccccccc)}
\mathbbm{I} &        &        &        &        &        &        \\
             & \ddots &        &        &        &        &        \\
             &        & \mathbbm{I} &        &        &        &        \\
             &        &        & \ddots &        &        &        \\
             &        &  M_{P_iP_j}  &   & \mathbbm{I} &        &        \\
             &        &        &        &        & \ddots &        \\
             &        &        &        &        &        & \mathbbm{I} \\
\end{block}
\end{blockarray},
 \]
and $\mathbbm{I}$ denotes the identity matrix of size $\frac{n}{k}\times \frac{n}{k}$.

Notice that each $D_t$ can implemented by a local CNOT circuit on $P_t$ with no nonlocal gates. By the construction in \Cref{lem:dagcomm} and the fact that $P_i$ holds $n/k$ input qubits, each $L_{i,j}$ can be implemented with $O(n/k)$ CNOTs between $P_i$ and $P_j$, using 1 ancilla qubit on $P_j$. Note that implementing one CNOT between $P_i$ and $P_j$ needs at most $D$ hops in topology graph $H$, which yields $O(D)$ nonlocal CNOTs. As there are $O(n^2)$ matrices $L_{i,j}$, $O(n/k\cdot D\cdot n^2)=O(nkD)$ nonlocal CNOTs are required in total.

\end{proof}


\subsection{Lower bound for Clifford Circuits}

We show that linear communication are necessary to implement Clifford circuits.
It suffice to prove that there exists an $n$-qubit CNOT circuit $T$ that
requires at least $\Omega(n)$ communication. Even stronger, we prove that it
holds for approximating $T$ up to constant error, as shown in the
following lemma.

\begin{lemma} \label{lem:cnotlow}
Given any constant $0\leq \epsilon < 1/4$, there exists an $n$-qubit CNOT circuit $T$ such that 
\[
\CC_\infty(T;\epsilon)=\Omega(n),
\] 
for all large enough $n$.
\end{lemma}

\begin{proof}

By Lemma~\ref{lem:matrix}, for all large enough $n$ there exists an invertible
binary matrix $M\in\mathbb F_2^{n\times n}$ with the property that every
$\frac{n}{2}\times \frac{n}{2}$ submatrix has rank at least $(1-\delta)n/2$
($\delta$ is specified later). Let $T$ be the CNOT circuit whose matrix
representation is $M$.

Fix an arbitrary balanced partition of the $n$ input qubits into two
processors $A$ and $B$, each holding $n/2+O(1)$ qubits and $m$ ancillas. Write $x=(x_A,x_B)$ 
and block-decompose $M$ according to the output on $(A,B)$ and
the input on $(A,B)$:
\[
M = \begin{pmatrix} M_{AA} & M_{AB}\\[2pt] M_{BA} & M_{BB}\end{pmatrix}.
\]
Consider the following experiment. Processor $A$ samples uniformally random 
$X_A\in\{0,1\}^{|A|}$ and prepares $\ket{X_A}$ on its input qubits,
while $B$ prepares $\ket{0^{|B|}}$, and all ancillas are set to $\ket{0}$. First
suppose $T$ is applied exactly. As $T$ maps $\ket{x}$ to $\ket{Mx}$, the final state is the classical
string $\ket{Y_A,Y_B}$, where
$Y_B = M_{BA} X_A$.  Observe that by performing Gaussian elimination on $Y_B$, $B$ is able to recover 
$\rank(M_{BA})$ uniformly random bits of $X_A$. Hence the mutual information
\[
I(A:B) = I(X_A:Y_B) = H(Y_B) \geq  \rank(M_{BA}).
\]
By the choice of $M$ and the balanced partition, $\rank(M_{BA})\ge (1-\delta)n/2-O(1)$. 
Thus, for exact implementation, $I(A:B)\geq(1-\delta)n/2-O(1)$.

Now suppose the implemented unitary $V$ $\epsilon$-approximates $T$ in spectral norm, and let $\rho'_{AB}$ be the resulting final
state. For the corresponding ideal final state $\rho_{AB}$ (obtained by $T$),
the input is pure, so for any such input $\ket{\psi}$ we have
$\|V\ket{\psi}-T\ket{\psi}\| \le \|V-T\|_{2}=\epsilon$, which
implies trace distance ${\rm TD}(\rho'_{AB},\rho_{AB})\le \epsilon$. By the
Fannes-Audenaert inequality (\Cref{lem:fannes}), for any system $Q\in\{A,B,AB\}$,
\[
|S(Q)_{\rho'}-S(Q)_{\rho}|
\le \epsilon\log(\dim Q-1) + H(\epsilon, 1-\epsilon)
\leq \epsilon \log \dim Q + 1
\]
where $H$ is the binary entropy. Therefore
\[
\left|I(A:B)_{\rho'} - I(A:B)_{\rho}\right|
\le |S(A)_{\rho'}-S(A)_{\rho}|+|S(B)_{\rho'}-S(B)_{\rho}|
+|S(AB)_{\rho'}-S(AB)_{\rho}|
\leq 2\epsilon n + 3.
\]
By the fact that $I(A:B)_{\rho} = (1-\delta)n/2-O(1)$, we have 
\[
I(A:B)_{\rho'} \geq 
I(A:B)_{\rho} - \left|I(A:B)_{\rho'} - I(A:B)_{\rho}\right| \geq
 \left(\frac{1-\delta}{2} - 2\epsilon\right)n-O(1).
\]
By setting $\delta=(1-4\epsilon)/2$, we have $I(A:B)_{\rho'} \geq (1/4-\epsilon)n-O(1)=\Omega(n)$ since constant $\epsilon<1/4$.

Finally, by \Cref{lem:info}, the number of nonlocal two-qubit gates used is
at least $I(A:B)_{\rho'}/4=\Omega(n)$. As the balanced partition and ancillas were arbitrary,
we have $\CC_m(T;\epsilon)=\Omega(n)$ regardless of $m$.
\end{proof}

\begin{lemma} \label{lem:matrix}
    Given any constant $\delta\in(0,1)$, for all large enough $n$, there exists an matrix $M\in \mathbb{F}_2^{n\times n}$ such that 
    (i) $M$ is invertible, and (ii) for any submatrix $N\in \mathbb{F}^{\frac{n}{2}\times \frac{n}{2}}$ of $M$, we have ${\rm rank}(N) \geq (1-\delta) n / 2$.
\end{lemma}

\begin{proof}
    Prove by randomly pick a $M\in \mathbb{F}^{n\times n}$. First, consider the probability of $M$ being invertible. The number of invertible $n\times n$ matrices over $\mathbb{F}$ is
    \(
    |{\rm GL}(n,2)|=\prod_{i=0}^{n-1} (2^n-2^i)
    \) 
    Thus we have
    \[
    \Pr[M \text{ invertible}]=
    \frac{|{\rm GL}(n,2)|}{2^{n^2}} = \prod_{i=1}^{n} (1-2^{-i}) = 
    \frac{1}{2}\prod_{i=2}^{n} \left(1-2^{-i}\right) \geq \frac{1}{2} \left(1-\sum_{i=2}^{n} 2^{-i}\right) > \frac{1}{4}.
    \]
    Next, a standard counting bound says
    \(\#\{A\in \mathbb F_2^{k\times k}:\ \operatorname{rank}(A)=r\}\le 2^{(2k-r)r}\). 
    Hence, for any fixed choice of submatrix $N\in \mathbb{F}^{\frac{n}{2}\times \frac{n}{2}}$, we have
    \[
    \Pr [{\rank}(N) < r] \leq \frac{1}{2^{n^2/4}}\sum_{t< r} 2^{(n-r)r} \leq (r+1) 2^{(n-r)r - n^2/4}.
    \]
    By setting $r=(1-\delta) n/2$, we have $\Pr [{\rank}(N) < r]\leq n
    2^{-\delta^2 n^2/4}$. Note that there are at most $\binom{n}{n/2}^2 \leq 2^{2n}$
    choices of $N$. By the union bound, the probability
    \[
    \Pr [\exists N: {\rank}(N) < (1-\delta)n/2] \leq 2^{2n} \cdot n 2^{-\delta^2 n^2/4} = 2^{-\Omega(\delta^2 n^2)}.
    \]
    Thus the probability that $M$ satisfies both (i) and (ii) is at least $1/4 - 2^{-\Omega(\delta^2 n^2)} > 0$ for large enough $n$, 
    which completes the proof.
\end{proof}

Finally, we conclude \Cref{thm:cnot}.
\begin{proof}[Proof of \Cref{thm:cnot}]
    Combine \Cref{coro:cliffordup} and \Cref{lem:cnotlow}.
\end{proof}

\section*{Acknowledgements}
We thank Ziheng Chen for helpful discussions. This work was supported by the Innovation Program for Quantum Science and Technology under Grant No.~2024ZD0300500.

\bibliography{ref}
\bibliographystyle{alpha}

\appendix
\section{Proof of \Cref{coro:commtopo}} \label{app:commtopo}

\begin{proof}

First consider the case with no ancillas ($m=0$). We design a synthesis
algorithm below that uses $O\left(4^{(1-1/k)n}\right)$ nonlocal two-qubit gates to
distribute $U$ on topology $G$. 

Fix a spanning tree $T$ of $G$ and repeat the following for $k-1$ rounds. Let
$\mathcal{U}_i$ be a set of $(n-(i-1)n/k)$-qubit unitaries with
$|\mathcal{U}_i|=4^{(i-1)n/k}$; initially $\mathcal{U}_1=\{U\}$. At round $i$, 
\begin{enumerate}
  \item Pick a leaf $v$ of $T$. For each $V\in\mathcal{U}_i$, apply \Cref{lem:decomponetopo} to
  decompose $V$ between $v$ and $T\setminus\{v\}$. This produces
  $|\mathcal{U}_i|\times 4^{n/k} = 4^{in/k}$ many $(n-in/k)$-qubit unitaries, forming
  $\mathcal{U}_{i+1}$; and it adds
  \begin{align} 
    \notag&|\mathcal{U}_i|\cdot\left(6\times 4^{n/k}+6\times 2^{n-(i-1)n/k}+12(k-i)\times 2^{(1-i/k)n}\right) \\
    \label{eq:topo}=&6\times 4^{in/k}+6\times 2^{n+(i-1)n/k}+12(k-i)\times 2^{n+(i-2)n/k}
  \end{align}
  nonlocal two-qubit gates. Compared to the clique topology, the third term in
  \eqref{eq:topo} is the extra overhead caused by routing on $T$.
  \item Remove $v$ from $T$.
\end{enumerate}
Summing the topology overhead (the third term in \eqref{eq:topo}) over $i=1$ to $k-1$ gives
\[
\sum_{i=1}^{k-1} 12(k-i)\times 2^{n+(i-2)n/k}
\le 48\left(2^{2n-3n/k}+k\cdot 2^{\,n-3n/k}\right)
=O\left(4^{(1-3/(2k))n}\right),
\]
which does not exceed $O\left(4^{(1-1/k)n}\right)$, so the
overall nonlocal gate count remains $O\left(4^{(1-1/k)n}\right)$.

Finally, extend to arbitrary $m\ge 0$ ancillas. As in the proof of \Cref{lem:upperk}, we
move the $n$ input qubits from the original $k$ processors onto the first
$k' = n/(n/k+m)$ processors. Moving each qubit along $G$ costs at most $D$ nonlocal SWAPs, so
the total moving cost is $O(Dn)$. Moreover, a similar analysis will show that the additional 
topology overhead of item (i) and (ii) in the proof of \Cref{lem:upperk} does not exceed
$O(4^{(1-1/k)n-m})$. This completes the proof.
\end{proof}

\begin{lemma}\label{lem:decomponetopo}
Let $k$ processors be connected by a tree $T$, and let $v$ be a leaf of $T$. Consider
implementing an $n$-qubit unitary $U$ over $T$, with each processor holding $n/k$ input qubits.
Then $U$ can be decomposed into
\begin{enumerate}
  \item $6\times 4^{n/k} + 6\times 2^{n} + 12(k-1)\times 2^{(1-1/k)n}$ two-qubit gates between 
        $v$ and $T\setminus\{v\}$; and
  \item $4^{n/k}$ unitaries on the remaining $k-1$ processors acting on $(n-n/k)$ qubits.
\end{enumerate}
\end{lemma}

\begin{proof}
Let $u$ be the neighbor of $v$. Starting from the qubits on $v$, apply the quantum Shannon
decomposition to $U$ a total of $n/k$ times to obtain:
\begin{enumerate}
  \item $4^{n/k}$ unitaries acting on the $(n-n/k)$ qubits in $T\setminus\{v\}$; and
  \item for each $0\le i<n/k$, a collection of $3\times 4^{i}$ UCRs on $n-i$ wires.
\end{enumerate}
For each UCR $R$ from item (2), first use two nonlocal SWAPs (at the beginning and end) to move
the target of $R$ from $\pi(t)$ to $u$, then apply \Cref{lem:ucdecomp} with $a=n/k-i$ to decompose
$R$ into CNOTs and $(n-n/k)$-qubit UCRs. In total each $R$ yields
\begin{enumerate}
  \item[(i)] 
  $\sum_{i=0}^{n/k-1} 3\times 4^i(2^{n/k-i}+2)\le 6\times 4^{n/k}$ two-qubit gates across
  between $v$ and $T\setminus\{v\}$; and
  \item[(ii)]
  $\sum_{i=0}^{n/k-1} 3\times 4^i 2^{n/k-i}\le 3\times 4^{n/k}$ UCRs acting only on
  $T\setminus\{v\}$.
\end{enumerate}
Then apply \Cref{lem:ucdisttopo} to distribute all UCRs from (ii) over the remaining $k-1$
processors, which requires
\[
3\times 4^{n/k}\times\left(2^{(1-2/k)n+1}+(k-1)\times 2^{(1-3/k)n+2}\right)
=6\times 2^{n} + 12(k-1)\times 2^{(1-1/k)n}
\]
additional nonlocal two-qubit gates. Therefore, after the above steps, $U$ is expressed as
$4^{n/k}$ unitaries on $T\setminus\{v\}$ each acting on $(1-1/k)n$ qubits, together with
\(
6\times 4^{n/k} + 6\times 2^{n} + 12(k-1)\times 2^{(1-1/k)n}
\)
nonlocal two-qubit gates, as claimed.
\end{proof}

\begin{lemma}\label{lem:ucdisttopo}
Let $k$ processors be connected by a tree $T$. Consider implementing an $n$-qubit UCR 
gate $R$ over $T$, with each processor holding $n/k$ input qubits. Then for any
qubit assignment $\pi$, $R$ can be implemented using
\(
2^{n-n/k+1} + (k-1)\times 2^{n-2n/k+2}
\)
nonlocal two-qubit gates.
\end{lemma}

\begin{proof}
Pick an arbitrary leaf $v$ of $T$ and let $u$ be its neighbor. Use SWAPs at the beginning and
end of the subcircuit to move the target qubit $t$ of $R$ from its current processor
$\pi(t)$ to $u$, which costs at most $2(k-1)$ nonlocal SWAPs. Then apply \Cref{lem:ucdecomp} to
decompose $R$, producing
\begin{enumerate}
  \item $2^{n/k}$ CNOTs between $v$ and $u$; and
  \item $2^{n/k}$ UCRs on the $(n-n/k)$ qubits in $T\setminus\{v\}$.
\end{enumerate}
Thus the cost so far is $2^{n/k}+2(k-1)$ nonlocal two-qubit gates.

Repeat this decomposition for $k-1$ rounds. The total nonlocal gate count is
\begin{align*}
\sum_{i=0}^{k-2} 2^{in/k}\left(2^{\,n/k}+2(k-1)\right)
= \frac{(2^n-2^{n/k})(2^{n/k}+2(k-1))}{4^{n/k}-2^{n/k}}
\le 2^{n-n/k+1} + (k-1)\times 2^{n-2n/k+2}.
\end{align*}
\end{proof}

\section{An Alternative Proof for Space-Bounded Lower Bound} \label{appx:simple}

\begin{lemma}[Weaker version of \Cref{cor:experror}] 
    Given integers $n > 0, m\geq 0$, and $k\geq 2$, there exists an $n$-qubit unitary $U$ such that
    \(\CC_m^{(k)}(U)=\Omega\left(
    4^{(1-1/k)n-m}\right).\)
\end{lemma}

\begin{proof}

Assume that every $n$-qubit unitary can be implemented on $k$ processors, each
holding $n/k$ input qubits and $m$ ancilla qubits, using only $\ell$ nonlocal
two-qubit gates. Let $p_i,q_i\in[k]$ denote the two processors participating in
the $i$-th nonlocal gate.
Fix a qubit partition $\pi:[n]\to [k]$ and a fixed sequence
$(p_1,q_1),\ldots,(p_\ell,q_\ell)$. Without loss of generality, we assume all
nonlocal gates to be CNOT gates and are applied to the first qubit of the
processors. Then any realizable unitary can be expressed as 
\begin{equation} \label{eq:ell}
    U_0 \cdot\mathrm{CNOT}_{p_1,q_1} \cdot U_1 \cdot \mathrm{CNOT}_{p_2,q_2} \cdots
U_{\ell-1} \cdot \mathrm{CNOT}_{p_\ell,q_\ell} \cdot U_\ell, 
\end{equation} 
where each $U_i$ is a tensor product of local unitaries and $\mathrm{CNOT}_{p_i,q_i}$
is a fixed nonlocal CNOT between the first qubit on $p_i, q_i$. Now count the
real parameters in $U_i$:
\begin{enumerate}
\item $U_0=\bigotimes_{j=1}^k U_0^{(j)}$ with $U_0^{(j)}\in {\rm SU}(2^{\,n/k+m+O(1)})$ on processor $j$,
so $U_0$ contributes $k\times 4^{\,n/k+m+O(1)}$ real parameters (ignoring additive $-1$’s in big-$O$).
\item For $i>0$, $U_i$ acts nontrivially only on processors $p_i$ and $q_i$, hence contributes
$2\times 4^{\,n/k+m+O(1)}$ real parameters.
\end{enumerate}
Each nonlocal CNOT has no free parameter. Therefore, the whole family obtained
with the fixed $(\pi,(p_i,q_i))$ contains at most $(k+2\ell)\times
4^{\,n/k+m+O(1)}$ real degrees of freedom.
On the other hand, ${\rm SU}(2^n)$ has real dimension $4^n-1$. By Sard's theorem
\cite{sard1942measure}, covering all of ${\rm SU}(2^n)$ requires
\[
(k+2\ell)\times 4^{\,n/k+m+O(1)}\ \ge\ 4^n-1
\Longrightarrow
\ell\ \ge\ 4^{\,(1-1/k)n-m-O(1)}\ =\ \Omega\!\left(4^{(1-1/k)n-m}\right).
\]
Otherwise, unitaries of the form Eq.~\eqref{eq:ell} is a measure-zero subset of ${\rm SU}(2^n)$.

Finally, the number of choices of $\pi$ and $(p_i,q_i)$ is finite. Since a
finite union of measure-zero sets is still measure zero, the lower bound on
$\ell$ holds even when allowing different choices of $\pi$ and $(p_i,q_i)$.
Therefore $\ell=\Omega\left(4^{(1-1/k)n-m}\right)$. 
\end{proof}

\section{Proof of \Cref{lem:rank}} \label{appx:rank}

\begin{proof}
    After exchanging $c$ qubits, the joint pure state of $A$ and $B$ before final measurement has Schmidt rank $\le 2^c$, which can be written as
    \[ \ket{\Psi}_{AB}=\sum_{i=1}^{2^{c}} \alpha_i\ket{\varphi_i}_{A}\otimes\ket{\psi_i}_{B}. \]
    Classical outputs $x, y$ are obtained by measuring $\ket{\Psi}_{AB}$ in computational basis. We have
\begin{align*}
M_{x,y}
&= \Pr(A\text{ outputs }x,\, B\text{ outputs }y)\\
&= \braket{x,y|\Psi}\,\braket{\Psi|x,y}\\
&= \bra{x,y}\left(\sum_{i,i'} {\alpha_i\alpha^\dagger_{i'}\ket{\varphi_i,\psi_i}\bra{\varphi_{i'},\psi_{i'}}}\right)\ket{x,y}\\
&= \sum_{i,i'} \alpha_i\alpha_{i'}^{\dagger}\;
   \braket{x|\varphi_i}\braket{\varphi_{i'}|x}\cdot
   \braket{y|\psi_i}\braket{\psi_{i'}|y}\,.
\end{align*}
Define vectors ${\bf p}_{i,i'}=(\braket{x|\varphi_i}\braket{\varphi_{i'}|x})_x$ and ${\bf q}_{i,i'}=(\braket{y|\psi_i}\braket{\psi_{i'}|y})_y$.
Then
$$
M=\sum_{1\leq i,i'\leq 2^c} M_i\quad\text{where }M_i:=\alpha_i\alpha_i^\dagger\ {\bf p}_{i,i'}^T\cdot {\bf q}_{i,i'}.
$$
By ${\rm rank}(M_i)=1$ and sub-additivity of matrix rank, we conclude ${\rm rank}(M)\leq 2^{2c}$, which implies $c\geq \frac{1}{2}\log {\rm rank}(M)$.
\end{proof}

\end{document}